\documentclass[11pt]{iopart}
\expandafter\let\csname equation*\endcsname\relax
\expandafter\let\csname endequation*\endcsname\relax
\usepackage{graphicx}
\usepackage{amssymb}
\usepackage{amsmath}
\usepackage{amsthm}
\usepackage{hyperref}
\usepackage{epstopdf}
\usepackage{tikz}
\usepackage{pstricks} 
\usepackage{comment}
\usepackage{tikz}
\usetikzlibrary{graphs,patterns,decorations.markings,arrows,matrix}

\newtheorem{thm}{Theorem}
\newtheorem{prop}{Proposition}
\newtheorem{lemma}{Lemma}

\usepackage[font=small, labelfont=bf, width=0.8\textwidth]{caption}

\DeclareMathOperator{\Span}{Span}
\DeclareMathOperator{\diag}{diag}

\def\be{\begin{equation}}
\def\ee{\end{equation}}

\newcommand{\ket}[1]{| #1 \rangle}

\tikzset{part/.style={circle,fill=black,inner sep=1.8pt,outer sep=3pt}}
\tikzset{hole/.style={thick,circle,draw=black,fill=white,inner sep=1.8pt,outer sep=3pt}}
\tikzset{bosline/.style={draw=blue!15!white,ultra thick}}
\tikzset{bossline/.style={draw=blue!50!white,ultra thick}}
\tikzset{bosssline/.style={draw=blue!90!white,ultra thick}}
\tikzset{partline/.style={draw=black!35!green,ultra thick}}
\tikzset{thinpart/.style={draw=black!35!green,thick}}
\tikzset{holeline/.style={draw=black!20!red,ultra thick}}
\tikzset{thinhole/.style={draw=black!35!red,thick}}
\tikzset{regionline/.style={draw=black!30!blue,ultra thick}}
\tikzset{thinreg/.style={draw=black!30!blue,thick}}
\tikzset{bgline/.style={dotted}}
\tikzset{bgplaq/.style={fill=lightgray!15!white}}
\tikzset{gplaq/.style={fill=green!20!white}}
\tikzset{arr/.style={postaction={decorate,thick,decoration={markings,mark = at position #1 with {\arrow{>}}}}}}
\tikzset{invarr/.style={postaction={decorate,thick,decoration={markings,mark = at position #1 with {\arrow{<}}}}}}
\tikzset{6varr/.style={postaction={decorate,thick,decoration={markings,mark = at position #1 with {\arrow{triangle 60}}}}}}
\tikzset{inv6varr/.style={postaction={decorate,thick,decoration={markings,mark = at position #1 with {\arrow{triangle 60 reversed}}}}}}
\newcommand\blank[2]
{
\draw[bgplaq] (#1,#2) -- (#1+1,#2) -- (#1+1,#2+1) -- (#1,#2+1) -- (#1,#2);
}
\newcommand\gblank[2]
{
\draw[gplaq] (#1,#2) -- (#1+1,#2) -- (#1+1,#2+1) -- (#1,#2+1) -- (#1,#2);
}
\newcommand\bosss[3]
{
\filldraw[white!10!blue, draw=black] (#1,#2) circle (#3cm);
}
\newcommand\boss[3]
{
\filldraw[white!50!blue, draw=black] (#1,#2) circle (#3cm);
}
\newcommand\bos[3]
{
\filldraw[white!85!blue, draw=black] (#1,#2) circle (#3cm);
}

\begin{document}

\title{Matrix product formula for Macdonald polynomials}
\author{${}^1$Luigi Cantini, ${}^2$Jan de Gier and ${}^3$Michael Wheeler}
\address{${}^1$Laboratoire de Physique Th\'eorique et Mod\'elisation (CNRS UMR 8089),
Universit\'e de Cergy-Pontoise, F-95302 Cergy-Pontoise, France, \\ ${}^{2,3}$Department of Mathematics and Statistics, The University of Melbourne, 3010 VIC, Australia}
\eads{${}^1$luigi.cantini@u-cergy.fr, ${}^2$jdgier@unimelb.edu.au, ${}^3$wheelerm@unimelb.edu.au}

\begin{abstract}
We derive a matrix product formula for symmetric Macdonald polynomials. Our results are obtained by constructing polynomial solutions of deformed Knizhnik--Zamolodchikov equations, which arise by considering representations of the Zamolodchikov--Faddeev and Yang--Baxter algebras in terms of $t$-deformed bosonic operators. These solutions are generalised probabilities for particle configurations of the multi-species asymmetric exclusion process, and form a basis of the ring of polynomials in $n$ variables whose elements are indexed by compositions. For weakly increasing compositions (anti-dominant weights), these basis elements coincide with non-symmetric Macdonald polynomials. Our formulas imply a natural combinatorial interpretation in terms of solvable lattice models. They also imply that normalisations of stationary states of multi-species exclusion processes are obtained as Macdonald polynomials at $q=1$.
\end{abstract}

\maketitle

\section{Introduction}

Symmetric Macdonald polynomials \cite{macd88,MacdBook} are families of multivariable orthogonal polynomials indexed by partitions whose coefficients depend rationally on two parameters $q$ and $t$. In the case $q=t$ they degenerate to the more familiar Schur functions which encode characters of irreducible representations of the symmetric group. One way to define Macdonald polynomials is as joint eigenfunctions of a family of commuting operators in the double affine Hecke algebra \cite{Cher95a,Cher95b}. They can also be defined combinatorially as generating functions \cite{HaglundHL1,HaglundHL2,RamY}, or via symmetrisation of non-symmetric Macdonald polynomials that are computed from Yang-Baxter graphs \cite{Lasc01,Lasc}.

Another family of commuting operators can be constructed in the Hecke algebra, corresponding to commuting transfer matrices in the theory of Yang--Baxter solvable lattice models \cite{baxterbook}. Those operators are solutions of the Yang--Baxter algebra generated by an $R$-matrix arising from a quantum group. Eigenfunctions of transfer matrices are generally complicated objects that can be constructed using the so-called Bethe ansatz. In cases where the eigenvalue is simple it is sometimes possible to construct explicit eigenfunctions of transfer matrices, and a well developed technique to construct such eigenfunctions is the matrix product algebra \cite{dehp} for the asymmetric exclusion process. In an inhomogeneous setting the matrix product algebra is known as the Zamolodchikov--Faddeev (ZF) algebra \cite{ZZ1979,Fad1980}. The latter algebra is related to the $t$-deformed Knizhnik--Zamolodchikov equation \cite{KZ,FR}. 

Opdam \cite{Opdam} and Cherednik \cite{Cher95a,Cher95b} generalised Macdonalds construction to a non-symmetric setting and defined families of non-symmetric Macdonald polynomials which are indexed by tuples of integers called compositions. There exists a basis in the ring spanned by non-symmetric Macdonald polynomials whose defining equations are exchange relations which are equivalent to the $t$-deformed Zamolodchikov--Faddeev algebra. Elements of this basis specialise to probabilities of particle configurations of the multi-species asymmetric exclusion process. Solutions to the ZF algebra can be obtained from solutions of the Yang--Baxter algebra, but not all solutions of the latter give non-trivial solutions of the former. Polynomial solutions of the deformed ZF and KZ algebras have received recent attention in the context of the Razumov--Stroganov--Cantini--Sportiello theorem and alternating-sign matrices \cite{DFZJ,Pasq,KasPasq,KasataniT,SigechiU,RSZJ,dGLS}.

In this paper we put several ingredients together to show that a solution of the Yang--Baxter algebra arising for the multi-species asymmetric exclusion process and based on $t$-deformed bosons generates non-trivial polynomial solutions of the Zamolodchikov--Faddeev algebra. We show that those solutions lead to matrix product formulas for certain basis functions in the ring of polynomials spanned by the non-symmetric Macdonald polynomials, and by symmetrisation this implies a matrix product formula for symmetric Macdonald polynomials. We furthermore provide a combinatorial interpretation of our formulas in terms of two-dimensional lattice configurations of boson trajectories. 

For $q=1$ the homogeneous limit of our formulas are steady state probabilities of multispecies asymmetric exclusion processes. In the context of the simple asymmetric exclusion process the ZF algebra was studied in \cite{SasaW97} and recently reviewed in a more general context of integrable exclusion processes in \cite{CrampeRV}. Our construction here is heavily based on recent results for the multispecies asymmetric exclusion process \cite{FerrariM1,FerrariM2,EvansFM,ProlhacEM,AritaAMP,AritaM}.

\subsection{Deformed bosons}
In the following we shall make extensive use of bosonic operators called $t$-deformed oscillators. These operators appear naturally in the context of solvable lattice models and in matrix product formulas for the asymmetric simple exclusion process (ASEP) \cite{SasaW97,BlytheECE}.  They are defined by
\be
\begin{split}
a k&=t k a,\qquad  a^\dag k = t^{-1}k a^\dag,\\
a a^\dag &= 1-k, \qquad a^\dag a = 1-t^{-1} k.
\end{split}
\label{eq:oscillators}
\ee
A consequence of these relations is the following equation which we shall use frequently,
\be
aa^\dag - t a^\dag a = 1-t.
\ee
A faithful representation of this algebra on the Fock space $\Span\{ \ket{m} \}_{m=0}^\infty$ is given by 
\be
\begin{split}
a\ket{m} &=  (1-t^{m})^{\frac12}\ket{m-1}, \qquad
a^\dagger \ket{m} =  (1-t^{m+1})^{\frac12}\ket{m+1}, \\
k&= \diag\{1,t,t^{2},\ldots\}.
\end{split}
\label{eq:oscillatorsrep}
\ee

\subsection{Hecke algebra action on polynomials}
The second algebraic ingredient we need is an action of the Hecke algebra on polynomials. We start by denoting by $S_{n}$ the symmetric group on $n$ elements, i.e. $S_{n}$ is the Weyl group of type $A_{n-1}$. The symmetric group has a natural action on  compositions $\lambda\in \mathbb{N}_0^n$ given by
\be
s_i(\ldots,\lambda_i,\lambda_{i+1},\ldots) = (\ldots,\lambda_{i+1},\lambda_i,\ldots).
\label{eq:symmgroup}
\ee
The symmetric group also has a natural action on polynomials, given by
\be
\label{eq:symmgroup_pol}
s_i f(\ldots,x_i,x_{i+1},\ldots) = f(\ldots,x_{i+1},x_{i},\ldots). 
\ee

A $t$-deformation of \eref{eq:symmgroup_pol} may be defined using divided differences leading to the Demazure operator
\be
T_i^{\pm 1} = t^{\pm 1/2} - t^{-1/2} \frac{t x_i - x_{i+1}}{x_i-x_{i+1}} (1-s_i),
\label{eq:Hecke}
\ee
which for $t\rightarrow 1$ reduces to \eref{eq:symmgroup_pol}. The operator $T_i$ satisfies the relations of the Hecke algebra of type $A_{n-1}$,
\begin{align}
(T_i-t^{1/2})(T_i+t^{-1/2})&=0,\qquad (i=1,\ldots,n-1)\nonumber\\
T_i T_{i\pm 1} T_i &= T_{i\pm 1} T_i T_{i\pm 1}, \\
T_iT_j &= T_jT_i\qquad |i-j|\ge 2. \nonumber
\end{align}
It is convenient to define the following shifted operator (sometimes referred to as Baxterised operator),
\begin{equation}
T_i(u)  = T_i +\frac{t^{-1/2}}{[u]},\qquad  [u]= \frac{1- t^{u}}{1-t}.
\label{eq:Demazure} 
\end{equation}
which satisfies the Yang--Baxter equation,
\begin{equation}
T_i(u)T_{i+1}(u+v)T_i(v) = T_{i+1}(v) T_i(u+v) T_{i+1}(u).
\end{equation}

We furthermore define the following shift operator compatible with the affine Hecke algebra of type $A_{n-1}$,
\begin{align}
(\omega f)(x_1,\ldots,x_n) &= f(qx_n,x_1,\ldots,x_{n-1}),\\
\omega T_i &= T_{i+1}\omega.
\end{align}
The affine Hecke algebra formed by the operators \eref{eq:Hecke} and $\omega$ has an Abelian subalgebra generated by the Murphy elements, which are defined as
\be
Y_i = T_i\cdots T_{n-1} \omega T_{1}^{-1} \cdots T_{i-1}^{-1}.
\label{eq:Yi}
\ee
These operators mutually commute and central elements in the Hecke algebra can be constructed by taking symmetric combinations of the Murphy elements. Joint eigenfunction of the operators $Y_i$ are non-symmetric Macdonald polynomials. 

For later use it is convenient to develop some notation. Representations of $S_n$ are indexed by partitions, which are those compositions for which $\lambda_1 \ge \lambda_2 \ge \ldots \ge \lambda_n$. Compositions are naturally ordered under the dominance order $\ge$ on $\mathbb{N}_0^n$, which is defined as
\be
\lambda\ge \mu\quad\text{if}\quad \sum_{i=1}^k (\lambda_i - \mu_i) \ge 0\quad \text{for }k=1,\ldots,n.
\ee

Let $\lambda$ be a composition and let $\lambda^+$ be the dominant weight of $\lambda$ in the dominance order on compositions, i.e. $\lambda^+$ is a permutation of $\lambda$ such that $\lambda_1 \geq \lambda_2 \geq \ldots \geq \lambda_n$. Let $w_+$ be the smallest word such that $\lambda = w_+\cdot\lambda^+$. The permutation $w_+^{-1}$ is obtained  by labeling each entry from $\lambda$ with a number from 1 to $n$, from the biggest entry to the smallest and from the left to the right. For instance $\lambda=(3,0,4,4,2) \Rightarrow w_+^{-1}=(3,5,1,2,4)$ and so $w_+=(3,4,1,5,2)$ and $\lambda^+=(4,4,3,2,0)$. The anti-dominant weight for this example is $\delta=(0,2,3,4,4)$. 

To be able to define the spectral vector in the next section we will need the quantity 
\be 
\rho(\lambda) := w_+\cdot\rho,
\ee 
where $\rho=\tfrac12(n-1,n-3,\ldots,-(n-1))$. For the example above, $n=5$ and therefore $\rho=(2,1,0,-1,-2)$ and $\rho(\lambda)=(0,-2,2,1,-1)$.

\subsection{Deformed Knizhnik--Zamolodchikov equations}

We are interested in finding explicit matrix product formulas for Macdonald polynomials. We first derive matrix product formulas for polynomials which are linear combinations of non-symmetric Macdonald polynomials. These polynonmials are defined by a system of exchange equations closely related to the $q$-deformed Knizhnik--Zamolodchkov ($q$KZ) equation. We note here that in our notation the parameter that is usual called $q$ is replaced by $t$ in order to make connection to the literature on Macdonald theory. 

Following Kasatani and Takeyama \cite{KasataniT}, polynomial solutions to the (reduced) $q$KZ equations, sometimes called exchange relations, can be obtained from eigenfunctions of the $Y_i$ operators. Let $\delta$ be the anti-dominant weight and let $E_\delta$ be the non-symmetric Macdonald polynomial solving the eigenvalue equation
\begin{equation}
Y_i E_\delta = y_i(\delta) E_\delta,
\end{equation}
where $y_i(\delta)=t^{\rho(\delta)_i} q^{\delta_i}$. We now define another set of polynomials, which are linear combinations of the Macdonald polynomials $E_\lambda$, and for which we will be able to find compact explicit expression. Define
\begin{align}
f_\delta &:= E_\delta,\nonumber\\
f_{\ldots,\lambda_i,\lambda_{i+1},\ldots} &:= t^{-1/2} T_i^{-1} f_{\ldots,\lambda_{i+1},\lambda_{i},\ldots} \quad \lambda_i > \lambda_{i+1}.
\end{align}
 Note that our notation differs slightly from \cite{KasataniT} as \eqref{qKZ2} contains a factor $t^{-1/2}$ on the right hand side, and consequently some other details below are different.

Then $f$ solves the $q$KZ equations
\begin{align}
T_i f_{\ldots,\lambda_i,\lambda_{i+1},\ldots} &= t^{1/2} f_{\ldots,\lambda_i,\lambda_{i+1},\ldots}\quad \lambda_i = \lambda_{i+1},\label{qKZ1}\\
T_i f_{\ldots,\lambda_i,\lambda_{i+1},\ldots} &= t^{-1/2} f_{\ldots,\lambda_{i+1},\lambda_i,\ldots}\quad \lambda_i > \lambda_{i+1},\label{qKZ2}\\
\omega f_{\lambda_n,\lambda_1,\ldots,\lambda_{n-1}} &= q^{\lambda_n} f_{\lambda_1,\ldots,\lambda_n}.\label{qKZ3}
\end{align}
%
%
Writing $q=t^{u}$ we define the elements of the \textit{spectral vector}  $\langle \lambda\rangle$ of a composition $\lambda$  as,
\be
\langle \lambda\rangle_i = \rho_i(\lambda)+ u \lambda_i.\qquad y_i(\lambda) = t^{\langle \lambda\rangle_i }.
\ee 
The non-symmetric Macdonald polynomials are obtained from $E_\delta$ by the action of Baxterised operators:
\begin{equation}
E_{s_i\lambda}=  T_i(\langle \lambda\rangle_{i+1} -\langle \lambda\rangle_{i} ) E_\lambda,\qquad \lambda_i < \lambda_{i+1}.
\label{eq:HonE}
\end{equation}
These polynomials satisfy
\begin{equation}
Y_i E_\lambda = y_i(\lambda) E_\lambda.
\end{equation}

Below we will derive a matrix product expression for the functions $f_\lambda$, and while we will not give an explicit matrix product expression for the polynomials $E_\lambda$, these can be derived as linear combinations of the $f_\lambda$ using \eref{eq:HonE}. In fact, the two families of polynomials are related via a triangular change of basis:
\begin{align}
\label{eq:triang}
E_{\lambda}
=
\sum_{\mu \leq \lambda}
c_{\lambda\mu}(q,t)
f_{\mu},
\qquad
f_{\lambda}
=
\sum_{\mu \leq \lambda}
d_{\lambda\mu}(q,t)
E_{\mu}
\end{align}
for suitable rational coefficients $c_{\lambda\mu}(q,t)$ and $d_{\lambda\mu}(q,t)$.

\section{Matrix Product, Yang--Baxter and Zamolodchikov--Faddeev algebras}

\subsection{Matrix product Ansatz}

The aim of this section is to obtain a matrix product expression for $E_{\delta}(x_1,\dots,x_n)$, the non-symmetric Macdonald polynomial indexed by the anti-dominant weight $\delta$. Our approach is to write an Ansatz for the polynomials $f_{\lambda}$, which generalize $E_{\delta}$, and to show that this Ansatz obeys the $q$KZ equations \eqref{qKZ1}--\eqref{qKZ3}. The Ansatz is as follows:
\begin{equation}
\Omega_{\lambda^+} f_{\lambda}(x_1,\ldots,x_n) = \Tr \Big[ A_{\lambda_1}(x_1) \cdots A_{\lambda_n}(x_n) S\Big],
\label{MPA}
\end{equation}
where $\Omega_{\lambda^+}$ is a normalisation factor to be determined later and $A_0(x),A_1(x),\dots,A_r(x)$ and $S$ are matrices satisfying the exchange relations
\begin{align}
A_i(x)A_i(y) 
&= 
A_i(y)A_i(x),
\label{Exchange0}
\\
t A_j(x)A_i(y)
-
\frac{tx-y}{x-y}
\Big( A_j(x)A_i(y) - A_j(y)A_i(x) \Big)
&= A_i(x) A_j(y),
\label{Exchange1}
\\
S A_i(q x) 
&= 
q^{i} A_i(x)S, 
\label{Exchange2}
\end{align}
for all $0 \leq i<j \leq r$. It is straightforward to demonstrate that the Ansatz \eqref{MPA} is a faithful solution to the $q$KZ relations \eqref{qKZ1}--\eqref{qKZ3}. The validity of \eqref{qKZ1} is ensured by \eqref{Exchange0}, since $f_{\lambda}(x_1,\dots,x_n)$ is symmetric in $x_i,x_{i+1}$ when $\lambda_i = \lambda_{i+1}$. Accordingly, the action of $(1-s_i)$ gives zero in this case, and we find that 
\begin{align}
T_i
f_{\lambda_1,\dots,\lambda_i,\lambda_{i+1},\dots,\lambda_n} 
(x_1,\dots,x_n)
=
t^{1/2} f_{\lambda_1,\dots,\lambda_{i+1},\lambda_{i},\dots,\lambda_n} 
(x_1,\dots,x_n)
\end{align}
by inspection. In a similar vein, when $\lambda_i > \lambda_{i+1}$, direct application of \eqref{Exchange1} allows us to conclude that
\begin{align}
T_i
f_{\lambda_1,\dots,\lambda_i,\lambda_{i+1},\dots,\lambda_n} 
(x_1,\dots,x_n)
=
t^{-1/2} f_{\lambda_1,\dots,\lambda_{i+1},\lambda_{i},\dots,\lambda_n} 
(x_1,\dots,x_n).
\end{align}
Finally, to prove \eqref{qKZ3}, we observe that
\begin{align}
\Omega_{\lambda^+}f_{\lambda_n,\lambda_1,\dots,\lambda_{n-1}}
(q x_n,x_1,\dots,x_{n-1})
&=
\Tr\left(
A_{\lambda_n}(q x_n)
A_{\lambda_1}(x_1)
\cdots
A_{\lambda_{n-1}}(x_{n-1})
S
\right),
\nonumber\\
&=
\Tr\left(
A_{\lambda_1}(x_1)
\cdots
A_{\lambda_{n-1}}(x_{n-1})
S
A_{\lambda_n}(q x_n)
\right)\nonumber\\
&=
q^{\lambda_n}
\Omega_{\lambda^+} f_{\lambda_1,\dots,\lambda_n} 
(x_1,\dots,x_n),
\end{align}
where we have used the cyclicity of the trace and the exchange relation \eqref{Exchange2} to reach the final equality.
 
\subsection{Zamolodchikov--Faddeev algebra}

Solutions to the relations \eqref{Exchange0}--\eqref{Exchange1} can be recovered from the Yang--Baxter algebra corresponding to the quantum group $U_{t^{1/2}}(A^{(1)}_r)$, or rather a twisted version of it \cite{Resh90}. For models based on $U_{t^{1/2}}(A^{(1)}_r)$, the $R$-matrix can be expressed in the form
\begin{align}
\check{R}^{(r)} (x,y)
=
\sum_{i=1}^{r+1}
E^{(ii)}
\otimes
E^{(ii)}
+
\frac{x - y}{t x - y}
\sum_{1 \leq i < j \leq r+1}
\Big(
t E^{(ij)}
\otimes
E^{(ji)}
+
E^{(ji)}
\otimes
E^{(ij)}
\Big)\nonumber
\\
\frac{t-1}{t x - y}
\sum_{1 \leq i < j \leq r+1}
\Big(
x
E^{(ii)}
\otimes
E^{(jj)}
+
y
E^{(jj)}
\otimes
E^{(ii)}
\Big)
\end{align}
where $E^{(ij)}$ denotes the elementary $(r+1) \times (r+1)$ matrix with a single non-zero entry 1 at position $(i,j)$. The intertwining equation, or Yang--Baxter algebra, for such models is given by
\begin{align}
\check{R}(x,y)\cdot\left [L(x)\otimes L(y)\right] 
= 
\left [L(y)\otimes L(x)\right] \cdot \check{R}(x,y),
\label{eq:YBAdef}
\end{align}
where we have suppressed the superscript $(r)$ and $L(x)=L^{(r)}(x)$ is an $(r+1) \times (r+1)$ operator-valued matrix. This algebra is well-studied and many solutions for $L(x)$ are known. In the next section we will provide a solution for which the elements of $L(x)$ are given in terms of $t$-deformed quantum oscillators.

The exchange relations \eqref{Exchange0}--\eqref{Exchange1}  are equivalent to the Zamolodchikov--Faddeev (ZF) algebra \cite{ZZ1979,Fad1980},
\be
\check{R}(x,y )\cdot\left [\mathbb{A}(x)\otimes \mathbb{A}(y)\right] 
= 
\left [\mathbb{A}(y)\otimes \mathbb{A}(x)\right] ,
\label{eq:ZFdef}
\ee
where again we suppress $(r)$ and $\mathbb{A}=\mathbb{A}^{(r)}(x)$ is an $(r+1)$-dimensional operator valued column vector given by
\be
\mathbb{A}^{(r)}(x) = (A_0(x),\ldots, A_r(x))^T.
\ee
Equation \eref{Exchange2} is instead rewritten as
\begin{equation}\label{eq:SZF}
S \mathbb{A}(qx) =q^{\sum_i iE^{(ii)}}\mathbb{A}(x) S,
\end{equation}
where the rank $r$ is again implicit, i.e. $\mathbb{A}=\mathbb{A}^{(r)}(x)$ and $S=S^{(r)}$.

We can construct solutions of \eqref{eq:ZFdef} by rank-reducing the Yang--Baxter algebra \eref{eq:YBAdef} in the following way. Assume a solution of the following modified $RLL$ relation
\begin{align}
\check{R}^{(r)}(x,y)\cdot\left [\tilde{L}(x)\otimes \tilde{L}(y)\right] 
= 
\left [\tilde{L}(y)\otimes \tilde{L}(x)\right] \cdot \check{R}^{(r-1)}(x,y),
\label{eq:YBAdef_lowrank}
\end{align}
in terms of an $(r+1) \times r$ operator-valued matrix $\tilde{L}(x)=\tilde{L}^{(r)}(x)$, and an operator $s=s^{(r)}$ that satisfies
\begin{equation}
s\tilde{L}(qx) =q^{\sum iE^{(ii)}}\tilde{L}(x) sq^{-\sum iE^{(ii)}},
\end{equation}
which in components just means
\begin{equation}
\label{eq:twist-comp}
s\tilde{L}_{ij}(qx) =q^{i-j}\tilde{L}_{ij}(x) s.
\end{equation}
Then
\begin{align}
\label{eq:nestedMPA}
\mathbb{A}^{(r)}(x) &= \tilde{L}^{(r)}(x)\cdot \tilde{L}^{(r-1)}(x)  \cdots \tilde{L}^{(1)}(x),
\\
S^{(r)} &= s^{(r)}\cdot s^{(r-1)} \cdots s^{(1)}
\label{eq:nestedTwist}
\end{align}
gives a solution to \eqref{eq:ZFdef} and  \eref{eq:SZF} provided that the operator entries of $\tilde{L}^{(a)}(x)$ commute with those of $\tilde{L}^{(b)}(y)$, for all $a \not=b$. The usual way to ensure this commutativity is to demand that the entries of $\tilde{L}^{(a)}$ act on same vector space $V_a$ while $\tilde{L}^{(b)}$ act on a different vector spaces $V_b$, and indeed we shall adopt this approach in the coming sections. We will show that solutions to \eref{eq:YBAdef_lowrank} can be obtained from the Yang--Baxter algebra \eref{eq:YBAdef} by trivialising the representation of one of the quantum oscillators, with the consequence of reducing the rank of $L^{(r)}(x)$ by one and thus giving rise to $\tilde{L}^{(r)}(x)$.

\section{Low rank examples}

\subsection{Rank 1 solution to ZF algebra} 

Before presenting the general construction we will first display some explicit examples. For convenience we define the following functions
\begin{align}
b^+  &= \displaystyle \frac{t(x-y)}{tx-y},  &  b^-  &= t^{-1} b^+ = \frac{x-y}{tx-y}, \nonumber \\
\\[-\baselineskip]
c^+ & =1-b^+ = \frac{y (t-1)}{tx-y}, & c^- &=1-b^- = \frac{x (t-1)}{tx-y}. \nonumber
\end{align}
Then for $r=1$, we can trivially solve \eref{eq:YBAdef_lowrank} for $\tilde{L}^{(1)}(x)$: 
\begin{align*}
\left(
\begin{array}{cc|cc}
1 & 0 & 0 & 0
\\
0 & c^- &  b^+ & 0
\\
\hline
0 & b^- &c^+ & 0
\\
0 & 0 & 0 & 1
\end{array}
\right)
\cdot
\left[
\left(
\begin{array}{c}
1 \\ x
\end{array}
\right)
\otimes
\left(
\begin{array}{c}
1 \\ y
\end{array}
\right)\right]
=
\left[
\left(
\begin{array}{c}
1 \\ y
\end{array}
\right)
\otimes
\left(
\begin{array}{c}
1 \\ x
\end{array}
\right)
\right].
\end{align*}
Hence we see that 
\be
\tilde{L}^{(1)}(x)=\mathbb{A}^{(1)}(x)=\begin{pmatrix} 1\\ x \end{pmatrix},
\ee
is a rank $1$ solution to \eref{eq:ZFdef}. The corresponding solution to the Yang--Baxter algebra  \eqref{eq:YBAdef} is equal to
\be
L^{(1)}(x)=\begin{pmatrix} 1 & a \\ x a^\dag & x \end{pmatrix},
\ee
where the operators $a$, $a^\dagger$ and $k$ satisfy the $t$-oscillator relations \eref{eq:oscillators}. We note that trivialising the $t$-oscillator by sending $a^\dagger,a\mapsto 1$ and $k\mapsto 0$, we reduce the rank, and thus obtain the  solution $\tilde{L}^{(1)}(x)$:
\be
\begin{pmatrix} 1 & a \\  x a^\dag & x \end{pmatrix} \mapsto \begin{pmatrix} 1 & 1 \\ x & x \end{pmatrix}.
\ee

\subsection{Rank 2 solution to ZF algebra}

The rank $2$ case is less trivial, giving rise to operator valued solutions for $\mathbb{A}^{(2)}(x)$. In the case $r=2$, solving equation \eref{eq:YBAdef_lowrank} for $\tilde{L}^{(2)}(x)$, we find that
\begin{align}
\left(
\begin{array}{ccc|ccc|ccc}
1 & 0 & 0 & 0 & 0 & 0 & 0 & 0 & 0
\\
0 & c^- & 0 & b^+ & 0 & 0 & 0 & 0 & 0
\\
0 & 0 &c^-& 0 & 0 & 0 & b^+ & 0 & 0
\\
\hline
0 & b^- & 0 & c^+& 0 & 0 & 0 & 0 & 0
\\
0 & 0 & 0 & 0 & 1 & 0 & 0 & 0 & 0
\\
0 & 0 & 0 & 0 & 0 & c^- & 0 & b^+& 0
\\
\hline
0 & 0 & b^- & 0 & 0 & 0 &  c^+ & 0 & 0
\\
0 & 0 & 0 & 0 & 0 & b^- & 0 & c^+& 0
\\
0 & 0 & 0 & 0 & 0 & 0 & 0 & 0 & 1
\end{array}
\right)\cdot
\left[\left(
\begin{array}{cc}
1 & a \\ x k & 0 \\ x a^\dag & x
\end{array}
\right)\otimes
\left(
\begin{array}{cc}
1 & a \\ y k& 0 \\ y a^\dag & y
\end{array}
\right)\right]
= \nonumber
\\
\left[\left(
\begin{array}{cc}
1 & a \\ y k  & 0\\ y a^\dag & y
\end{array}
\right)\otimes
\left(
\begin{array}{cc}
1 & a \\ x k & 0 \\ x a^\dag & x
\end{array}
\right)\right]\cdot
\left(
\begin{array}{cc|cc}
1 & 0 & 0 & 0
\\
0 & c^- & b^+& 0
\\
\hline
0 & b^- & c^+ & 0
\\
0 & 0 & 0 & 1
\end{array}
\right),
\end{align}
Using \eref{eq:nestedMPA} we thus construct a solution of the ZF algebra in the following way:
\be
\mathbb{A}^{(2)}(x) = \tilde{L}^{(2)}(x)\cdot \tilde{L}^{(1)}(x) = \begin{pmatrix} 1 & a \\ x k & 0 \\ x a^\dag & x \end{pmatrix}  \begin{pmatrix} 1  \\ x \end{pmatrix} =  \begin{pmatrix} 1  + x a \\ k x \\ x a^\dag+x^2 \end{pmatrix} .
\label{eq:ZFr=3}
\ee
The associated rank 2 solution to the Yang--Baxter algebra is
\be
L^{(2)}(x)=
\begin{pmatrix}
1 & a_1 & a_2 \\
x a_1^\dag k_2 & x k_2 & 0 \label{eq:osc}\\
x a_2^\dag & x a_1 a_2^\dag & x
\end{pmatrix},
\ee
where $\{a_1,a^{\dag}_1,k_1\}$ and  $\{a_2,a^{\dag}_2,k_2\}$ are two commuting copies of the $t$-oscillator algebra \eref{eq:osc}. The map $a_1^\dagger,a_1\mapsto 1$ and $k_1\mapsto 0$ reduces the rank of $L^{(2)}(x)$ by one
\be
L^{(2)}(x) \mapsto
\begin{pmatrix}
1 & 1 & a_2 \\
x k_2  & x k_2 & 0\\
x a_2^\dag & x a_2^\dag & x
\end{pmatrix}\  \Rightarrow\ 
\tilde{L}^{(2)}(x)=
\begin{pmatrix}
1 &  a_2 \\
x k_2 & 0\\
x a_2^\dag & x
\end{pmatrix},
\label{eq:TildeL2}
\ee
where the indices of $t$-oscillators are redundant in the final matrix, since we no longer need to distinguish between the two copies of the algebra.

We note that low-rank examples of $L$-matrices based on deformed oscillators have been treated in earlier literature, for example, \cite{SasaW97,BogoB92,BogoIK,BoosGKNR,KunibaOS}.

\subsection{A polynomial example}
\label{se:polexample}
We look at an explicit example for rank 2 taking $\delta=(0,0,1,1,2,2)$. In the notation of \cite{KasataniT} this leads to $\rho(\delta)=\frac12(-3,-5,1,-1,5,3)$. The nonsymmetric Macdonald polynomial that generates solutions to the $q$KZ equation satisfies the following equations:
\be
\begin{array}{ll}
Y_1 E_\delta = t^{-3/2} E_\delta \qquad & Y_4 E_\delta = q t^{-1} E_\delta \\
Y_2 E_\delta = t^{-5/2} E_\delta & Y_5 E_\delta = q^2 t^{5/2} E_\delta\\
Y_3 E_\delta = q t E_\delta & Y_6 E_\delta = q^2 t^{3/2} E_\delta.
\end{array}
\label{eq:Edefr=2}
\ee
Using the notation 
\begin{equation}
q=t^{u},\qquad [m]=\frac{1-t^{m}}{1-t},
\end{equation}
it can be verified that the following polynomial solves \eref{eq:Edefr=2},
\begin{multline}
E_\delta(x_1,\ldots,x_6;q,t)= x_3 x_4 x_5^2 x_6^2 + \frac{t^2}{[3+u]} (x_1+x_2)x_3x_4 x_5x_6(x_5+x_6) \\
+ \frac{t^{4}[2]}{[3+u][4+u]}x_1x_2x_3x_4x_5x_6.
\label{eq:Edeltar=2}
\end{multline}

We now verify the matrix product form 
\be
\Omega_{\delta^+} E_\delta(x_1,\ldots,x_6;q=t^{u},t)=\Tr \big[ A_0(x_1)A_0(x_2)A_1(x_3)A_1(x_4)A_2(x_5)A_2(x_6)S\big],
\ee
for this explicit solution. From \eref{eq:ZFr=3} we see that
\begin{align}
A_0(x) &= 1+xa,\nonumber \\
A_1(x) &= x k,\\
A_2(x) &= x a^\dag + x^2,\nonumber
\end{align}
and using \eref{Exchange2} we note that $S$ should satisfy
\be
qSa - a S =0, \qquad Sa^\dag -q a^\dag S =0.
\ee
Taking the explicit representation \eref{eq:oscillatorsrep} for the $t$-oscillators, $S$ has the form
\be
S = k^{u} = \diag\{1,t^{u},t^{2u},\ldots\} = \diag\{1,q,q^{2},\ldots\}.
\ee
Up to a normalisation, the nonsymmetric Macdonald polynomial $E_\delta$ is now represented in matrix product form by
\begin{align}
&\Tr \left[ \left(1+ x_1 a\right) \left(1+ x_2 a\right) x_3 k x_4 k  x_5 \left(a^\dag+ x_5 \right) x_6 \left(a^\dag+ x_6 \right)S\right] \nonumber\\
&=x_3x_4x_5x_6 \Tr  \left[\left( x_5 x_6 k^2 + (x_1+x_2)(x_5+x_6)a k^2 a^\dag + x_1x_2 a^2 k^2 (a^\dag)^2 \right)S\right],
\end{align}
where other terms involving unequal powers of $a$ and $a^\dag$ have zero trace.
Normalising with $\Omega_{\delta^+}=\Omega_{221100}=\Tr(k^2S)$ we finally get
\begin{multline}
E_\delta(x_1,\ldots,x_6;q=t^{u},t) = x_3x_4x_5^2x_6^2 + x_3x_4x_5x_6(x_1+x_2)(x_5+x_6) t^2 
\frac{\Tr a a^\dag k^2S}{\Tr k^2S} \\ + x_1x_2x_3x_4 x_5x_6 t^4 \frac{\Tr a^2 (a^\dagger)^2 k^2S}{\Tr k^2S},
\label{eq:Edtrace}
\end{multline}
which can be shown to equal \eref{eq:Edeltar=2}. We give the details of calculating the traces in Section~\ref{se:traces}.

\section{Matrix product for general rank}

The ZF algebra for the case $x=y=1$ is known as the {\it matrix product algebra} for the multi-species asymmetric exclusion process. A general rank solution for this case, first for $t=\infty$ and then in terms of $t$-oscillators, was recently obtained by several authors in a sequence of works \cite{FerrariM1,FerrariM2,EvansFM,ProlhacEM, AritaAMP}. A generalisation of these results that includes a spectral parameter which can be chosen inhomogeneously was found earlier in \cite{InoueKO} and independently in the case of super-algebras in \cite{Tsuboi}. In Section~\ref{sec:graph} we shall illuminate the very natural combinatorial structure of this solution. We note that a different inhomogeneous generalisation of the multi-secpies ASEP, with species-dependent hopping parameters, was studied in \cite{LamW, AritaM,AyyerL12,AyyerL14}.  
\begin{thm}\label{th:L}
Consider a matrix $L^{(r)}(x)$ whose entries are given by
\begin{align}
L^{(r)}_{ij}(x)
=
\left\{
\begin{array}{ll}
x \prod_{m=i+1}^{r}
k_m,
&
i=j
\\
\\
x a_j a^\dag_i
\prod_{m=i+1}^{r} k_m,
&
i>j
\\
\\
0,
&
i<j
\end{array}
\right.
\label{eq:Losc1}
\end{align}
for all $1 \leq i,j \leq r$, and
\begin{align}
L^{(r)}_{0j}
=
a_j,\
1 \leq j \leq r,
\quad\quad
L^{(r)}_{i0}(x)
=
x a^\dag_i \prod_{m=i+1}^{r}
k_m,\
1 \leq i \leq r,
\quad\quad
L^{(r)}_{00}
=
1,
\label{eq:Losc2}
\end{align}
where $\{a_i,a^{\dag}_i,k_i\}$, $1 \leq i \leq r$ are $r$ commuting copies of the $t$-oscillator algebra \eref{eq:osc}. Then this $L$ matrix satisfies the intertwining equation \eqref{eq:YBAdef}.
\end{thm}
\noindent
An example of of $L$ for $r=3$ is given in \eref{eq:Lrank3}.
A proof of Theorem~\ref{th:L} can be obtained by a long brute force check of the intertwining equations (\ref{eq:YBAdef}), by distinguishing many different cases.
Here we present a more compact and elegant proof. First we start with an easy general property of the solutions of the intertwining equations (\ref{eq:YBAdef}).
\begin{lemma}\label{general-lemma1}
Let $L(x)$ be a $(r+1)\times (r+1)$ matrix of operators depending on a spectral parameter $x$, and $\{v_0,\dots,v_r\}$ and $\{u_0,\dots,u_r\}$ two constant vectors of commuting operators that also commute with the entries of $L(x)$. Define the matrix
$$
{\mathcal L}(x) = \diag\{v_0,\dots,v_r\}\cdot L(x)\cdot \diag \{u_0,\dots,u_r\}.
$$
\begin{enumerate}
\item If $L(x)$ is a solution of the intertwining equations (\ref{eq:YBAdef}), then ${\mathcal L}(x)$ solves (\ref{eq:YBAdef}).
\item Conversely, if ${\mathcal L}(x)$ is a solution of the intertwining equations (\ref{eq:YBAdef}) and the operators $v_i,u_j$ are non identically zero, then $L(x)$ solves (\ref{eq:YBAdef}).
\end{enumerate}
\end{lemma}
Next we want to describe how to get a solution of the intertwining equations (\ref{eq:YBAdef}) of rank $r$ from a solution of rank $r-1$. 
For $0 \leq m \leq r-1$ introduce an insertion operator
$
\mathcal I_{m,\kappa}:\textrm{Mat}(r)\rightarrow \textrm{Mat}(r+1)
$
\begin{equation}
\mathcal I_{m,\kappa}( M)_{ij}= \left\{
\begin{array}{lll}
M_{ij} & 0\leq i\leq m-1, & 0\leq j \leq r-1,\\
\kappa \delta_{j,r} & i=m, & 0\leq j \leq r,\\
M_{i-1,j} &  m+1\leq i \leq r, & 0\leq j\leq r-1,\\
0 & 0\le i \le m & j=r. 
\end{array}
\right.
\end{equation}
Here is an example for $r=3$,
\be
g= 
\left(
\begin{array}{ccc}
g_{00} & g_{01} & g_{02}\\
g_{10} & g_{11} & g_{12}\\
g_{20} & g_{21} & g_{22}
\end{array}
\right) ~~~\longrightarrow ~~~
\mathcal I_{1,\kappa}(g)= 
\left(
\begin{array}{cccc}
g_{00} & g_{01} & g_{02} & 0 \\
0 & 0 & 0 & \kappa\\
g_{10} & g_{11} & g_{12} & 0\\
g_{20} & g_{21} & g_{22} & 0
\end{array}
\right)
\ee
\begin{lemma}\label{general-lemma2}
Let $L(x)$ be a solution of the intertwining equations (\ref{eq:YBAdef}) of rank $r-1$ such that the first $m$ rows are independent of $x$, while the last $r-m$ are linear in $x$. 
The matrix $\mathcal I_{m,\kappa}\left(L(x)\right)$ is a solution of the intertwining equations (\ref{eq:YBAdef}) of rank $r$ if and only if the operator $\kappa$ satisfies the following commutations
\begin{equation}
\begin{array}{lll}
\kappa L_{ij}(x)= L_{ij}(x)\kappa & \forall j, &  0\leq i\leq m-1\\
\kappa L_{ij}(x)= tL_{ij}(x)\kappa & \forall j, & m\leq i\leq r-1
\end{array}
\end{equation}
\end{lemma}
\begin{proof}
The proof of this lemma is a straightforward check of (\ref{eq:YBAdef}) .
\end{proof}
We have already seen that the matrix 
\be
L(a,a^\dagger,x)= \left(
\begin{array}{cc}
1 & a\\
xa^\dag & x\end{array}
\right)
\ee
satisfies  the intertwining equations (\ref{eq:YBAdef}) of rank~$1$, and we set 
\be
L^{(1)}(x) = L(a_1,a_1^\dagger,x).
\ee
Now we are going to construct $L^{(r)}(x)$ starting from $L^{(r-1)}(x)$
in a recursive way.
%
The first  ingredient we need is the auxiliary matrix 
\be
{\mathcal L}^{(+,r)}(x)=\mathcal I_{0,P_0}\left(\diag\{1,Q_0,\dots,Q_0\}\cdot L^{(r-1)}(x)\right)
\ee
where $P_0,Q_0$ commute with the entries of $ L^{(r-1)}(x)$ and satisfy
\be
P_0Q_0=tQ_0P_0.
\ee
Assuming that $ L^{(r-1)}(x)$ satisfies  the intertwining
equations and applying Lemma~\ref{general-lemma1} and Lemma~\ref{general-lemma2} it follows that ${\mathcal L}^{(+,r)}(x)$ satisfies the intertwining equations (\ref{eq:YBAdef}).  

The second ingredient we need is the matrix $\bar{\mathcal L}^{(+,r)}_{i,j}(x)$ defined by
\begin{equation}
\bar{\mathcal L}^{(+,r)}(x)= \mathcal I_{r-1,P_1k_r} \cdots \mathcal
I_{1,P_1k_r}\left(L(a_r,a_r^\dagger, x)\cdot\diag\{1,Q_1\}\right)
\end{equation}
where here $P_1,Q_1$ commute with $a_r,a_r^\dagger,k_r$ and $t$-commute among themselves
$$
P_1Q_1=tQ_1P_1.
$$
By applying Lemma~\ref{general-lemma1} and repeatedly
Lemma~\ref{general-lemma2} we have that $\bar {\mathcal L}^{(+,r)}(x)$ also satisfies the intertwining equations (\ref{eq:YBAdef}). 

Then we put the two ingredients together and define 
\begin{equation}
\begin{split}
{\mathcal L}^{(r+1)}(x) :=  \bar {\mathcal L}^{(+,r)}(x)\cdot{\mathcal L}^{(+,r)}(x),
\end{split}
\end{equation}
which obviously satisfies the intertwining equations (\ref{eq:YBAdef}) if $L^{(r)}(x)$ does.
Comparing ${\mathcal L}^{(r+1)}(x)$ and $L^{(r+1)}(x)$, we see that
\begin{equation}\label{barL-L}
{\mathcal L}^{(r+1)}(x)= \diag\{1,Q_0P_1,Q_0P_1,\dots,Q_0P_1,1 \}\cdot L^{(r+1)}(x)\cdot \diag\{1,\dots,1,Q_1P_0\},
\end{equation}
therefore Theorem~\ref{th:L} follows  from the second point of Lemma~\ref{general-lemma1} and induction on $r$.

\subsection{Graphical interpretation of $L$-matrix}
\label{sec:graph}

While the $L$-matrix as given by equations \eqref{eq:Losc1} and \eqref{eq:Losc2} solves the intertwining equation \eqref{eq:YBAdef}, the precise form of the entries is quite mysterious without any further explanation. The aim of this section is to show that all entries can be deduced from some governing combinatorial rules. To that end, we denote the entries of the $L$-matrix graphically by tiles. The left and right edges of the tiles are either unoccupied (corresponding with index 0), or occupied by a boson of colour $i$ (corresponding with index $i$). For example, for rank 3, the entries of the $L$-matrix are encoded as follows:
\begin{align}
L^{(3)}(x)
=
\left(
\begin{array}{cccc}
\begin{tikzpicture}[scale=0.6]
\blank{0}{0};
\end{tikzpicture}
&
\begin{tikzpicture}[scale=0.6]
\blank{0}{0};
\bos{1}{0.5}{0.1};
\end{tikzpicture}
&
\begin{tikzpicture}[scale=0.6]
\blank{0}{0};
\boss{1}{0.5}{0.1};
\end{tikzpicture}
&
\begin{tikzpicture}[scale=0.6]
\blank{0}{0}{0.1};
\bosss{1}{0.5}{0.1};
\end{tikzpicture}
\\
\begin{tikzpicture}[scale=0.6]
\blank{0}{0};
\bos{0}{0.5}{0.1};
\end{tikzpicture}
&
\begin{tikzpicture}[scale=0.6]
\blank{0}{0};
\bos{0}{0.5}{0.1};
\bos{1}{0.5}{0.1};
\end{tikzpicture}
&
\begin{tikzpicture}[scale=0.6]
\blank{0}{0};
\bos{0}{0.5}{0.1};
\boss{1}{0.5}{0.1};
\end{tikzpicture}
&
\begin{tikzpicture}[scale=0.6]
\blank{0}{0};
\bos{0}{0.5}{0.1};
\bosss{1}{0.5}{0.1};
\end{tikzpicture}
\\
\begin{tikzpicture}[scale=0.6]
\blank{0}{0};
\boss{0}{0.5}{0.1};
\end{tikzpicture}
&
\begin{tikzpicture}[scale=0.6]
\blank{0}{0};
\boss{0}{0.5}{0.1};
\bos{1}{0.5}{0.1};
\end{tikzpicture}
&
\begin{tikzpicture}[scale=0.6]
\blank{0}{0};
\boss{0}{0.5}{0.1};
\boss{1}{0.5}{0.1};
\end{tikzpicture}
&
\begin{tikzpicture}[scale=0.6]
\blank{0}{0};
\boss{0}{0.5}{0.1};
\bosss{1}{0.5}{0.1};
\end{tikzpicture}
\\
\begin{tikzpicture}[scale=0.6]
\blank{0}{0};
\bosss{0}{0.5}{0.1};
\end{tikzpicture}
&
\begin{tikzpicture}[scale=0.6]
\blank{0}{0};
\bosss{0}{0.5}{0.1};
\bos{1}{0.5}{0.1};
\end{tikzpicture}
&
\begin{tikzpicture}[scale=0.6]
\blank{0}{0};
\bosss{0}{0.5}{0.1};
\boss{1}{0.5}{0.1};
\end{tikzpicture}
&
\begin{tikzpicture}[scale=0.6]
\blank{0}{0};
\bosss{0}{0.5}{0.1};
\bosss{1}{0.5}{0.1};
\end{tikzpicture}
\end{array}
\right)
=
\left(
\begin{array}{cccc}
1 & a_1  & a_2  & a_3
\\
xk_3 k_2 a_1^\dag & xk_3 k_2 & 0 & 0
\\
xk_3 a_2^\dag & xk_3 a_2^\dag a_1 & xk_3 & 0
\\
x a_3^\dag & x a_3^\dag a_1 & x a_3^\dag a_2 & x
\end{array}
\right)
\label{eq:Lrank3}
\end{align}
where the indexing conventions are $\begin{tikzpicture} \bos{0}{0}{0.08} \end{tikzpicture} = 1$, $\begin{tikzpicture} \boss{0}{0}{0.08} \end{tikzpicture} = 2$, $\begin{tikzpicture} \bosss{0}{0}{0.08} \end{tikzpicture} = 3$, with an empty edge representing the index 0. 

This information alone is sufficient to label the entries of the $L$-matrix unambiguously. But for the purpose of motivating the form of the entries it is useful to consider, further to this, lattice paths which are generated by the bosons. The horizontal edges of tiles can be occupied by arbitrary numbers of bosons from each of the families, and we adopt the convention that the bosons are ordered from darkest to lightest, reading from left to right. In passing from the bottom horizontal edge to the top one, the occupation number of any species of boson can go up/down by 1 (representing the action of a creation/annihilation operator), or remain the same (representing the action of some diagonal operator). We keep track of these transitions between states by simply ``connecting the dots''. For example,

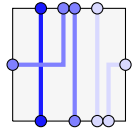
\begin{align}
\label{fig:tile1}
\begin{tikzpicture}[scale=1.5,baseline=(bas.base)]
\node (bas) at (0,0.5) {};
\blank{0}{0};
\draw[bosssline] (0.25,0) -- (0.25,1);
\draw[bossline] (0.55,0) -- (0.55,1);
\draw[bossline] (0,0.5) -- (0.45,0.5) -- (0.45,1);
\draw[bosline] (0.85,0) -- (0.85,0.5) -- (1,0.5);
\draw[bosline] (0.75,0) -- (0.75,1);
\bosss{0.25}{0}{0.05};
\boss{0.55}{0}{0.05};
\bos{0.75}{0}{0.05};
\bos{0.85}{0}{0.05};
\bosss{0.25}{1}{0.05};
\boss{0.45}{1}{0.05};
\boss{0.55}{1}{0.05};
\bos{0.75}{1}{0.05};
\boss{0}{0.5}{0.05};
\bos{1}{0.5}{0.05};
\end{tikzpicture}
\end{align}
represents the action of the operator $a^{\dag}_2 a_1$ (creation of a type 2 boson, annihilation of a type 1, with time flow up the page). These graphical conventions suffice to explain all creation/annhilation operators appearing in the $L$-matrix \eqref{eq:Lrank3}.

In order to specify the zero entries of \eqref{eq:Lrank3}, as well as the inclusion of the $k_i$ operators, we need two rules for the lattice paths thus constructed:
\begin{enumerate}
\item[\textbf{1.}] It is forbidden to create a boson of type $i$ and simultaneously annihilate one of type $j$, when $i<j$.

\item[\textbf{2.}] We obtain a factor of $t$ every time a type $j$ line horizontally crosses a type $i$ line, where $i>j$.
\end{enumerate}
Taking into account these rules, we can now reproduce the exact form of all entries of the $L$-matrix. The zero entries correspond with tiles forbidden by rule {\bf 1}. Rule {\bf 2} explains which entries are dressed by $k_i$ operators. For example, \eqref{fig:tile1} corresponds to $L^{(3)}_{2,1}= k_3 a_2^\dag a_1$, with $k_3$ producing a factor of $t^{m_3}$, where $m_3$ is the number of type 3 bosons present. We easily deduce this from rule {\bf 2}, since we will have exactly $m_3$ horizontal crossings of type 3 lines by the left-turning type 2 line. A further example:
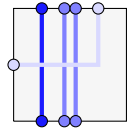
\begin{center}
\begin{tikzpicture}[scale=1.5]
\blank{0}{0};
\draw[bosssline] (0.25,0) -- (0.25,1);
\draw[bossline] (0.45,0) -- (0.45,1);
\draw[bossline] (0.55,0) -- (0.55,1);
\draw[bosline] (0,0.5) -- (0.75,0.5) -- (0.75,1);
\bosss{0.25}{0}{0.05};
\boss{0.45}{0}{0.05};
\boss{0.55}{0}{0.05};
\bosss{0.25}{1}{0.05};
\boss{0.45}{1}{0.05};
\boss{0.55}{1}{0.05};
\bos{0.75}{1}{0.05};
\bos{0}{0.5}{0.05};
\end{tikzpicture}
\end{center}
corresponds with $L^{(3)}_{1,0} = k_3 k_2 a^\dag_1$, since the left-turning type 1 boson horizontally crosses all type 2 and 3 lines present, and we must keep track of these crossings by the inclusion of $k_3 k_2$. 

\subsection{Trivialising a representation}
\label{sec:triv}

To obtain a solution to  the rank-reduced intertwining equation \eref{eq:YBAdef_lowrank} we choose the representation of the first $t$-oscillator algebra to be trivial:
\begin{align*}
a_1 = a^{\dag}_1 = 1,
\quad\quad
k_1 = 0.
\end{align*}
With this choice of representation, it is easy to see that the first two columns of the $L$-matrix \eqref{eq:Losc1}, \eqref{eq:Losc2} become equal, leading to an immediate solution of \eref{eq:YBAdef_lowrank} (by simply omitting one of the redundant columns from the $L$-matrix). For the example \eref{eq:Lrank3}, after trivializing the first $t$-oscillator algebra, we obtain
\begin{align}
\tilde{L}^{(3)}(x)
=
\left(
\begin{array}{cccc}
\begin{tikzpicture}[scale=0.6]
\blank{0}{0};
\end{tikzpicture}
&
\begin{tikzpicture}[scale=0.6]
\blank{0}{0};
\boss{1}{0.5}{0.1};
\end{tikzpicture}
&
\begin{tikzpicture}[scale=0.6]
\blank{0}{0};
\bosss{1}{0.5}{0.1};
\end{tikzpicture}
\\
\begin{tikzpicture}[scale=0.6]
\blank{0}{0};
\bos{0}{0.5}{0.1};
\end{tikzpicture}
&
\begin{tikzpicture}[scale=0.6]
\blank{0}{0};
\bos{0}{0.5}{0.1};
\boss{1}{0.5}{0.1};
\end{tikzpicture}
&
\begin{tikzpicture}[scale=0.6]
\blank{0}{0};
\bos{0}{0.5}{0.1};
\bosss{1}{0.5}{0.1};
\end{tikzpicture}
\\
\begin{tikzpicture}[scale=0.6]
\blank{0}{0};
\boss{0}{0.5}{0.1};
\end{tikzpicture}
&
\begin{tikzpicture}[scale=0.6]
\blank{0}{0};
\boss{0}{0.5}{0.1};
\boss{1}{0.5}{0.1};
\end{tikzpicture}
&
\begin{tikzpicture}[scale=0.6]
\blank{0}{0};
\boss{0}{0.5}{0.1};
\bosss{1}{0.5}{0.1};
\end{tikzpicture}
\\
\begin{tikzpicture}[scale=0.6]
\blank{0}{0};
\bosss{0}{0.5}{0.1};
\end{tikzpicture}
&
\begin{tikzpicture}[scale=0.6]
\blank{0}{0};
\bosss{0}{0.5}{0.1};
\boss{1}{0.5}{0.1};
\end{tikzpicture}
&
\begin{tikzpicture}[scale=0.6]
\blank{0}{0};
\bosss{0}{0.5}{0.1};
\bosss{1}{0.5}{0.1};
\end{tikzpicture}
\end{array}
\right)
=
\left(
\begin{array}{cccc}
1  & a_2   & a_3  
\\
x k_3 k_2  & 0 & 0
\\
x k_3 a_2^\dag & x k_3 & 0
\\
x a_3^\dag  & x a_3^\dag a_2 & x
\end{array}
\right).
\label{eq:TildeL}
\end{align}
Up to some elementary transformations, this is exactly the rank 3 solution obtained in \cite{ProlhacEM} (see equation (47) therein), although here we write our operators with subscripts to distinguish commuting copies of the $t$-oscillator algebra, rather than the tensor product notation employed in \cite{ProlhacEM}. 

\subsection{Solution of ZF algebra}

Having obtained solutions of \eqref{eq:YBAdef_lowrank} for all $r$, we construct solutions of the ZF algebra via the prescription \eqref{eq:nestedMPA}. This is conceptually straightforward, although it introduces a slight notational complexity: we must now distinguish not only between different families of $t$-bosons (which we have done so far by using subscripts), but also between operators which act at different lattice sites (which we now do by placing superscripts on our operators). For example the rank 3 solution of the ZF algebra is given by
\be
\label{rank3ZF}
\mathbb{A}^{(3)}(x) = 
\left(
\begin{array}{ccc}
1  & a_2   & a_3  
\\
x k_3 k_2  & 0 & 0
\\
x k_3 a_2^\dag & x k_3 & 0
\\
x a_3^\dag  & x a_3^\dag a_2 & x
\end{array}
\right)^{(3)}
\cdot
\left(
\begin{array}{cc}
1  & a_2    
\\
x k_2  & 0 
\\
x a_2^\dag   & x
\end{array}
\right)^{(2)}
\cdot
\begin{pmatrix}
1  \\
x
\end{pmatrix}^{(1)}
=
\begin{pmatrix}
A_0(x)\\
A_1(x)\\
A_2(x)\\
A_3(x)
\end{pmatrix}.
\ee
where the superscript placed on a matrix indicates that all operators within that matrix acquire that superscript. In terms of the graphical conventions that we have introduced, the components $A_i(x)$ of the rank $r$ solution $\mathbb{A}^{(r)}(x)$ are given by rows of tiles of length $r$. The left boundary of the row is occupied by a particle of colour $i$, for $1 \leq i \leq r$, or unoccupied in the case $i=0$. The right boundary is always unoccupied. For example, formulating \eqref{rank3ZF} in terms of tiles, we obtain
\be
\label{rank3ZFtiles}
\mathbb{A}^{(3)}(x) 
= 
\left(
\begin{array}{ccc}
\begin{tikzpicture}[scale=0.6]
\blank{0}{0};
\end{tikzpicture}
&
\begin{tikzpicture}[scale=0.6]
\blank{0}{0};
\boss{1}{0.5}{0.1};
\end{tikzpicture}
&
\begin{tikzpicture}[scale=0.6]
\blank{0}{0};
\bosss{1}{0.5}{0.1};
\end{tikzpicture}
\\
\begin{tikzpicture}[scale=0.6]
\blank{0}{0};
\bos{0}{0.5}{0.1};
\end{tikzpicture}
&
\begin{tikzpicture}[scale=0.6]
\blank{0}{0};
\bos{0}{0.5}{0.1};
\boss{1}{0.5}{0.1};
\end{tikzpicture}
&
\begin{tikzpicture}[scale=0.6]
\blank{0}{0};
\bos{0}{0.5}{0.1};
\bosss{1}{0.5}{0.1};
\end{tikzpicture}
\\
\begin{tikzpicture}[scale=0.6]
\blank{0}{0};
\boss{0}{0.5}{0.1};
\end{tikzpicture}
&
\begin{tikzpicture}[scale=0.6]
\blank{0}{0};
\boss{0}{0.5}{0.1};
\boss{1}{0.5}{0.1};
\end{tikzpicture}
&
\begin{tikzpicture}[scale=0.6]
\blank{0}{0};
\boss{0}{0.5}{0.1};
\bosss{1}{0.5}{0.1};
\end{tikzpicture}
\\
\begin{tikzpicture}[scale=0.6]
\blank{0}{0};
\bosss{0}{0.5}{0.1};
\end{tikzpicture}
&
\begin{tikzpicture}[scale=0.6]
\blank{0}{0};
\bosss{0}{0.5}{0.1};
\boss{1}{0.5}{0.1};
\end{tikzpicture}
&
\begin{tikzpicture}[scale=0.6]
\blank{0}{0};
\bosss{0}{0.5}{0.1};
\bosss{1}{0.5}{0.1};
\end{tikzpicture}
\end{array}
\right)^{(3)}
\cdot
\left(
\begin{array}{cc}
\begin{tikzpicture}[scale=0.6]
\blank{0}{0};
\end{tikzpicture}
&
\begin{tikzpicture}[scale=0.6]
\blank{0}{0};
\bosss{1}{0.5}{0.1};
\end{tikzpicture}
\\
\begin{tikzpicture}[scale=0.6]
\blank{0}{0};
\boss{0}{0.5}{0.1};
\end{tikzpicture}
&
\begin{tikzpicture}[scale=0.6]
\blank{0}{0};
\boss{0}{0.5}{0.1};
\bosss{1}{0.5}{0.1};
\end{tikzpicture}
\\
\begin{tikzpicture}[scale=0.6]
\blank{0}{0};
\bosss{0}{0.5}{0.1};
\end{tikzpicture}
&
\begin{tikzpicture}[scale=0.6]
\blank{0}{0};
\bosss{0}{0.5}{0.1};
\bosss{1}{0.5}{0.1};
\end{tikzpicture}
\end{array}
\right)^{(2)}
\cdot
\left(
\begin{array}{c}
\begin{tikzpicture}[scale=0.6]
\blank{0}{0};
\end{tikzpicture}
\\
\begin{tikzpicture}[scale=0.6]
\blank{0}{0};
\bosss{0}{0.5}{0.1};
\end{tikzpicture}
\end{array}
\right)^{(1)}
=
\begin{pmatrix}
A_0(x)\\
A_1(x)\\
A_2(x)\\
A_3(x)
\end{pmatrix}.
\ee
From this it is easy to extract individual components, for example:
\begin{align*}
A_2(x)
=
\begin{array}{ccccc}
\begin{tikzpicture}[scale=0.8,baseline=(bas.base)]
\node (bas) at (0,0.3) {};
\foreach\x in {3,...,1}{
\blank{\x}{0};
\node at (4.5-\x,-0.5) {\color{red}\tiny (\x)};
}
\draw[bossline] (1,0.5) -- (1.5,0.5) -- (1.5,1);
\boss{1.5}{1}{0.1};
\boss{1}{0.5}{0.1};
\end{tikzpicture}
&
\quad
+
&
\begin{tikzpicture}[scale=0.8,baseline=(bas.base)]
\node (bas) at (0,0.3) {};
\foreach\x in {3,...,1}{
\blank{\x}{0};
\node at (4.5-\x,-0.5) {\color{red}\tiny (\x)};
}
\draw[bossline] (1,0.5) -- (2.5,0.5) -- (2.5,1);
\boss{1}{0.5}{0.1};
\boss{2}{0.5}{0.1};
\boss{2.5}{1}{0.1};
\end{tikzpicture}
&
\quad
+
&
\begin{tikzpicture}[scale=0.8,baseline=(bas.base)]
\node (bas) at (0,0.3) {};
\foreach\x in {3,...,1}{
\blank{\x}{0};
\node at (4.5-\x,-0.5) {\color{red}\tiny (\x)};
}
\draw[bossline] (1,0.5) -- (1.5,0.5) -- (1.5,1);
\draw[bosssline] (2.5,0) -- (2.5,0.5) -- (3.5,0.5) -- (3.5,1);
\boss{1}{0.5}{0.1};
\boss{1.5}{1}{0.1};
\bosss{2.5}{0}{0.1};
\bosss{3}{0.5}{0.1};
\bosss{3.5}{1}{0.1};
\end{tikzpicture}
\\
\quad
x k^{(3)}_3 {a_2^\dag}^{(3)}
&
&
\quad
x^2 k^{(3)}_3 k^{(2)}_2
&
&
\quad
x^2 k^{(3)}_3 {a_2^\dag}^{(3)} a_2^{(2)}
\end{array}
\end{align*}
where the top and bottom edges of the row can be occupied by arbitrary numbers of ``background'' particles, which we do not show here, since they play no role in specifying the operators at each site.

In the discussion in Sections \ref{sec:graph} and \ref{sec:triv} we made no distinction between the colour of a particle, and the family of bosons that it represents. Indeed, a particle of colour $i$ corresponded with the $t$-oscillator algebra $\{a_i,a^{\dag}_i,k_i\}$.  However (with the conventions that we have adopted) in passing to the full solution of the ZF algebra, which is valued on $V^{(r)} \otimes \cdots \otimes V^{(1)}$,  the notions of {\it colour} and {\it family} no longer coincide across all $V^{(j)}$. For example in \eqref{rank3ZFtiles}, a particle of colour \begin{tikzpicture} \bosss{0}{0}{0.08} \end{tikzpicture} corresponds with the algebra $\{a_3,a^{\dag}_3,k_3\}$ when in column 3, with $\{a_2,a^{\dag}_2,k_2\}$ when in column 2, and with the trivialized algebra $\{a_1,a^{\dag}_1,k_1\}$ when in column 1. The general rule is that a particle of colour $i$ in column $r-j$ corresponds with the $t$-oscillator family $i-j$, and this shift should be kept in mind in the subsequent sections.

\subsection{Representation of twist operator}

Having constructed solutions of \eqref{eq:YBAdef_lowrank} for all values of $r$, we now seek an explicit operator $s$ which satisfies \eqref{eq:twist-comp}. This is very easily achieved by considering the form of the entries of $\tilde{L}(x)$. We let $s$ be factorized over the non-trivial copies of the $t$-oscillator algebra, with the following commutation relations:
\begin{align}
\label{eq:twist-factor}
s=s_r \cdots s_2,
\quad
\quad
a_i s_i = q^{i-1} s_i a_i,
\quad
\quad
a_i^\dag s_i = q^{1-i} s_i a_i^\dag,
\quad
\quad
k_i s_i = s_i k_i, 
\end{align}
where as usual all operators carry the superscript $(r)$, which we have suppressed for visual clarity.

\begin{prop}
Equation \eqref{eq:twist-comp} is satisfied with $s$ as defined by \eqref{eq:twist-factor}. 
\end{prop}

We obtain a valid solution of \eqref{eq:twist-factor} by letting
\begin{align}
s_i = k_i^{(i-1)u} = {\rm diag}\{1,q^{(i-1)},q^{2(i-1)},\dots\}_i
\end{align}
which is the representation that we use in all subsequent sections.

\section{Transition formulas}

The aim of this section is to write a recursion relation for the function $f_{\lambda}(x_1,\dots,x_n)$ for any composition $\lambda$, using the fact that both $\mathbb{A}^{(r)}(x)$ and $S^{(r)}$ are recursively defined. Indeed, we can rewrite equations \eqref{eq:nestedMPA} and \eqref{eq:nestedTwist} as
\begin{align}
\label{A^N}
\mathbb{A}^{(r)}(x) 
&=
\tilde{L}^{(r)}(x)\mathbb{A}^{(r-1)}(x)
\\
\label{S^N}
S^{(r)} 
&=  
s^{(r)}S^{(r-1)}
\end{align}
where the action in the vector space $V^{(r)}$ is explicitly factorized out. This suggests the definition of a \emph{transition matrix:}
\begin{align*}
T_{\lambda,\mu}(x_1,\dots,x_n) 
:= 
\Tr\left[
\tilde{L}_{\lambda_1,\mu_1}(x_1)
\cdots 
\tilde{L}_{\lambda_n,\mu_n}(x_n) 
s \right],
\end{align*}
where all operators carry a subscript $(r)$ which we have suppressed, and $\lambda,\mu$ are compositions whose parts lie in $[0,1,\dots,r]$ and $[0,1,\dots,r-1]$, respectively. We will use this transition matrix to derive the recursion relation that is of interest to us, but before we do that, we introduce a simple operation on compositions. Let $\lambda^{*}$ be the composition obtained by subtracting 1 from all non-zero parts of $\lambda$:
\begin{align*}
\lambda^{*}_i = \max(\lambda_i-1,0).
\end{align*}

\begin{thm}
Let $\lambda$ be a composition whose parts satisfy $0\leq \lambda_i\leq r$, for all $1 \leq i \leq n$. The following recursion relation holds:
\begin{align}
\label{trans-eq}
f_{\lambda}(x_1,\dots,x_n)
=
\prod_{i=1}^{r-1} \left(1- q^{i} \prod_{j \leq i} t^{m_j(\lambda)} \right)
\sum_{\mu\in S_n\cdot (\lambda^{*})^{+}}
T_{\lambda,\mu}(x_1,\dots,x_n) f_{\mu}(x_1,\dots,x_n),
\end{align}
where $m_i(\lambda)$ is the number of parts of size $i$ in $\lambda$ and the sum is taken over all compositions $\mu$ that are a permutation of $\lambda^{*}$.
\end{thm}
\begin{proof}
Rewrite (\ref{MPA}) as 
\begin{align*}
\Omega_{\lambda^+}f_{\lambda}(x_1,\dots,x_n) 
= 
\Tr \left[ A_{\lambda_1}(x_1) \cdots A_{\lambda_n}(x_n) S \right],
\label{MPA2}
\end{align*}
for some proportionality factor $\Omega_{\lambda^+} = \Omega^{(r)}_{\lambda^+}$ that
remains to be determined (and in particular needs to be shown different from zero). We remark that the factor $\Omega^{(r)}_{\lambda^+}$ is the same for all partitions $\lambda$ with the common re-ordering $\lambda^{+}$. Combining (\ref{MPA2}) with equations \eqref{A^N} and \eqref{S^N}, we obtain
\begin{multline}
\Omega^{(r)}_{\lambda^+}f_{\lambda}(x_1,\dots,x_n)
=
\\ 
\sum_{\mu}
\Tr \left[ 
\tilde{L}^{(r)}_{\lambda_1,\mu_1}(x_1) 
\cdots 
\tilde{L}^{(r)}_{\lambda_n,\mu_n}(x_n) 
s^{(r)}
\right]
\Tr \left[ 
A^{(r-1)}_{\mu_1}(x_1) \cdots A^{(r-1)}_{\mu_n}(x_n) S^{(r-1)}
\right],
\end{multline}
which can be rewritten as
\begin{align}
\label{trans-eq1}
\Omega^{(r)}_{\lambda^+}
f_{\lambda}(x_1,\dots,x_n)
=
\sum_{\mu}
T_{\lambda,\mu}(x_1,\dots,x_n) 
\Omega^{(r-1)}_{\mu^+} 
f_{\mu}(x_1,\dots,x_n).
\end{align}
At this stage the sum in \eqref{trans-eq1} is taken over all compositions $\mu$, with parts in $[0,1,\dots,r-1]$. In order to reduce the size of this sum, we need the following result.
\begin{prop}
For any two compositions $\lambda$ and $\mu$, if $\mu^{+} \neq (\lambda^{*})^{+}$ then
$T_{\lambda,\mu}(x_1,\dots,x_n)=0$.
\end{prop}
\begin{proof}
Given a composition $\lambda$, let $m_i(\lambda)$ be the number of its parts 
of size $i$. We wish to show that if $m_{i}(\lambda) \neq m_{i-1}(\mu)$ for some $i \geq 2$, then 
$T_{\lambda,\mu}(x_1,\dots,x_n)=0$. 

For each $i \geq 2$, let $\#_i(\lambda,\mu)$ be the number of pairs $(\lambda_k,\mu_k)=(i,i-1)$. Invoking the explicit form of $\tilde{L}(x)$, we see that the number of $a_i^\dag$ operators appearing in $T_{\lambda,\mu}(x_1,\dots,x_n)$ is equal to $m_i(\lambda)-\#_i(\lambda,\mu)$. Conversely, the number of $a_i$ operators appearing in this trace is given by 
$m_{i-1}(\mu)-\#_i(\lambda,\mu)$. In order for the trace to be non-vanishing, these two numbers must be equal for all $i \geq 2$.
\end{proof}
An immediate consequence of this proposition is that we can rewrite \eref{trans-eq1} as
\begin{align}
\label{trans-eq2}
\Omega^{(r)}_{\lambda^+}
f_{\lambda}(x_1,\dots,x_n)
=
\Omega^{(r-1)}_{(\lambda^*)^+}
\sum_{\mu\in S_n\cdot (\lambda^{*})^{+}}
T_{\lambda,\mu}(x_1,\dots,x_n) 
f_{\mu}(x_1,\dots,x_n) 
\end{align}
To complete the proof of \eqref{trans-eq}, it remains only to demonstrate the following result.
\begin{prop}
The proportionality factors appearing in \eqref{trans-eq2} satisfy
\begin{align}
\label{prop2}
\Omega^{(r)}_{\lambda^+}
=
\Omega^{(r-1)}_{(\lambda^*)^+}
\prod_{i=1}^{r-1} \frac{1}{1- q^{i}  t^{\lambda'_{1}-\lambda'_{i+1}}},
\end{align}
where $\lambda'$ is the conjugate of $\lambda^+$.
\end{prop}
\begin{proof}
Without losing generality, we ease notation by assuming that $\lambda$ is a partition. We compare coefficients of the monomial $x^{\lambda} := x_1^{\lambda_1}\cdots x_n^{\lambda_n}$ on both sides of \eref{trans-eq2}. Since the matrix elements $T_{\lambda,\mu}(x_1,\dots,x_n)$ are at most degree 1 in each $x_i$, the only contribution to $x^\lambda$ in the right hand side comes from $\mu=\lambda^{*}$. Therefore we can write
\begin{align}
\label{propor-trans}
\Omega^{(r)}_{\lambda}
x^{\lambda} 
=
\Omega^{(r-1)}_{\lambda^{*}}
T_{\lambda,\lambda^{*}}(x_1,\dots,x_n)
x^{\lambda^{*}}
=
\Omega^{(r-1)}_{\lambda^{*}}
T_{\lambda,\lambda^{*}}(1,\dots,1)
x^{\lambda},
\end{align}
or more simply, $\Omega^{(r)}_{\lambda} = \Omega^{(r-1)}_{\lambda^{*}} 
T_{\lambda,\lambda^{*}}(1,\dots,1)$. A simple, explicit computation gives
\begin{align*}
T_{\lambda,\lambda^{*}}(1,\dots,1)
= 
\prod_{i=1}^{r-1} \Tr\left[ k^{m_1(\lambda)} \cdots k^{m_i(\lambda)} k^{i u} \right] 
= 
\prod_{i=1}^{r-1} \frac{1}{1- t^{m_1(\lambda)} \cdots t^{m_i(\lambda)} q^{i}} .
\end{align*}
\end{proof}
Thanks to this proposition, \eqref{trans-eq2} implies \eqref{trans-eq} provided that 
$\Omega^{(r-1)}_{\lambda^{*}}$ is non-vanishing. This can be deduced inductively on $r$, using \eqref{prop2} with $\Omega^{(0)} = 1$. Indeed, we find that
\begin{align*}
\Omega^{(r)}_{\lambda}
=
\prod_{j=1}^{r}\prod_{i=1}^{r-j} 
\frac{1}{1- q^{i} \prod_{j \leq k < i+j} t^{m_k(\lambda)}}
=
\prod_{1 \leq i<j \leq r}
\frac{1}{1- q^{j-i} t^{\lambda'_{i}-\lambda'_{j}}}.
\end{align*}
\end{proof}
A simple case of \eqref{trans-eq} is when all parts of the composition $\lambda$ are non-zero, $\lambda_i > 0$ for all $1 \leq i \leq n$. In this case the summation on the right hand side is trivial, and we recover 
\begin{align*}
f_{\lambda}(x_1,\dots,x_n) 
= 
\left(
\prod_{i=1}^{n} 
x_i
\right)
f_{\lambda^{*}}(x_1,\dots,x_n).
\end{align*}

\section{Combinatorial interpretation of matrix product formula}

\subsection{Lattice formulation of matrix product}

In view of the graphical representation of the elements of the $\tilde{L}$-matrices, it is possible to interpret the matrix product formula \eqref{MPA} entirely in terms of lattice paths. This leads to a combinatorial rule for $f_{\lambda_1,\dots,\lambda_n}(x_1,\dots,x_n)$, for any composition $\lambda$.

We consider a lattice formed by $n$ rows of tiles, with associated spectral parameters $x_1,\dots,x_n$, where the index of the spectral parameters increases as we go from the top row to the bottom. Each row is $r$ tiles wide, where $r$ is the size of the largest part in $\lambda$. For each 
$\lambda_i > 0$, there is a boson of colour $i$ incident on the left boundary of the $i^{\rm th}$ row. However, the rows $i$ for which $\lambda_i = 0$ have no particle at the left boundary. The right boundary of the lattice is completely unoccupied by particles. The lattice thus constructed reproduces all terms in the matrix product \eqref{MPA}, since each row corresponds with one of the components $A_i(x)$ of $\mathbb{A}^{(r)}(x)$. For example, for $r=3$ and $\lambda =(0,2,3,1,0,2)$, the matrix product can be represented in the following way:
\begin{align*}
{\rm Tr}(A_0(x_1) A_2(x_2) A_3(x_3) A_1(x_4) A_0(x_5) A_2(x_6) S)
=
\begin{tikzpicture}[scale=0.9,baseline=(bas.base)]
\node (bas) at (0,3.5) {};
\foreach\x in {3,...,1}{
\foreach\y in {6,...,1}{
\blank{\x}{\y};
}
\node at (4.5-\x,0.5) {\color{red} \tiny (\x) };
}
\foreach\y in {6,...,1}{
\node at (0,-\y+7.5) {$x_{\y}$};
}
\boss{1}{1.5}{0.1};
\bos{1}{3.5}{0.1};
\bosss{1}{4.5}{0.1};
\boss{1}{5.5}{0.1};
\bos{1.5}{7}{0.1};
\boss{2.4}{7}{0.1};
\boss{2.6}{7}{0.1};
\bosss{3.5}{7}{0.1};
\end{tikzpicture}
\end{align*}

Some care needs to be taken in regard to the boundary conditions on the top and bottom edges of the lattice. Since we are taking a trace in \eqref{MPA}, the correct way to view the lattice is with the top and bottom edges identified, so that lattice paths are permitted to ``wrap around'' the cylinder thus formed. This ensures that in each column of the lattice, there is conservation between the number of particles entering/leaving that column (in other words, all creation/annihilation operators come in pairs, ensuring a non-vanishing trace). The sole exception to this rule involves the boson species whose representation has been trivialized. Recall that in column $j$ of the lattice, bosons from any of the families 
$\{1,\dots,j\}$ are allowed, but we trivialize the representation of the first family. Hence at the top of each column we are permitted to eject type 1 particles without causing the trace to vanish. In the example above, we eject 1 trivial boson from column 3, 2 trivial bosons from column 2, and 1 trivial boson from column 1, since these are the corresponding multiplicities in the composition $\lambda$.

Having set up the boundary conditions of the lattice in this way, we obtain all possible terms in the matrix product \eqref{MPA} by summing over all lattice paths which evolve from the left edge to the top boundary (with round-the-cylinder contributions being allowed). Note that each configuration gives rise to an explicit product of operators, whose trace must then be calculated, rather than local Boltzmann weights, as would normally be the case in combinatorial calculations of this nature.

Let us remark finally that equation \eqref{trans-eq} has a particularly natural meaning in the above combinatorial framework. If we represent $f_{\lambda}(x_1,\dots,x_n)$ in terms of an $n \times r$ lattice as described, then \eqref{trans-eq} can be interpreted as summing over all possible configurations of column $r$ of the lattice, with the remaining $r-1$ columns giving rise to $f_{\mu}(x_1,\dots,x_n)$. The elements $T_{\lambda,\mu}(x_1,\dots,x_n)$ of the transition matrix can thus be viewed as partition functions of a single column. We illustrate the calculation of $T_{\lambda,\mu}(x_1,\dots,x_n)$ on an explicit example in the Appendix.

\subsection{Rank 2 example}

We give an example of the lattice-path calculation for $f_{\delta}(x_1,\dots,x_n)$, where 
$\delta = (0,0,1,1,2,2)$ is an anti-dominant weight. This is of particular interest, since in this case $f_{\delta} = E_{\delta}(x_1,\dots,x_n;q,t)$. The possible lattice configurations for this example:

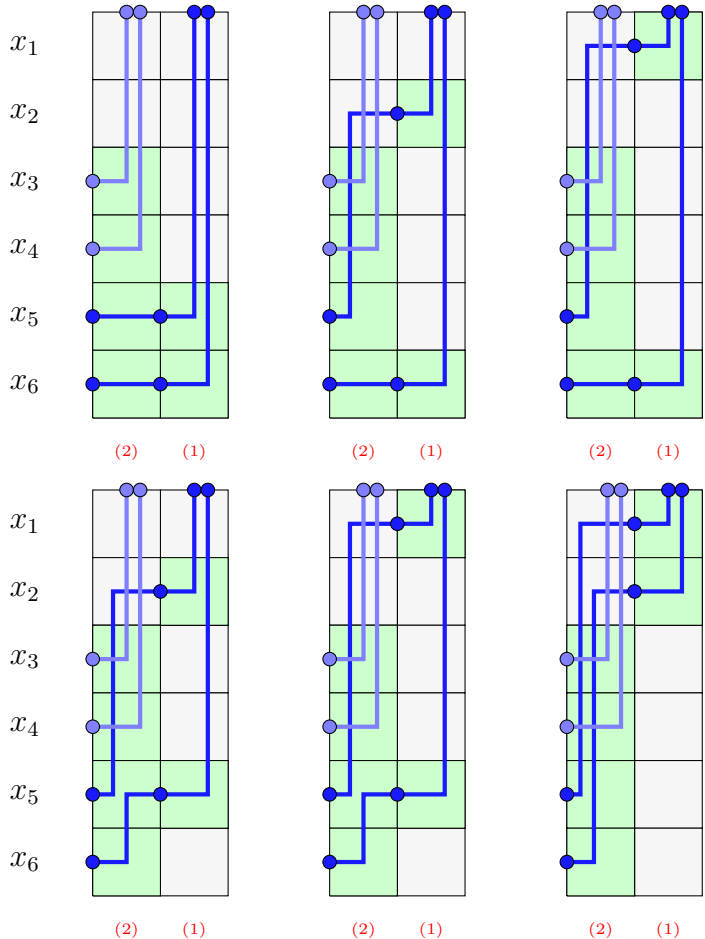
\begin{center}
\begin{tabular}{ccc}
\begin{tikzpicture}[scale=0.9]
\foreach\x in {2,...,1}{
\foreach\y in {6,...,1}{
\blank{\x}{\y};
}
\node at (3.5-\x,0.5) {\color{red} \tiny(\x)};
}
\foreach\y in {6,...,1}{
\node at (0,-\y+7.5) {$x_{\y}$};
}
\gblank{1}{4};
\gblank{1}{3};
\gblank{1}{2};
\gblank{1}{1};
\gblank{2}{2};
\gblank{2}{1};
\draw[bosssline] (1,1.5) -- (2.7,1.5) -- (2.7,7);
\draw[bosssline] (1,2.5) -- (2.5,2.5) -- (2.5,7);
\draw[bossline] (1,3.5) -- (1.7,3.5) -- (1.7,7);
\draw[bossline] (1,4.5) -- (1.5,4.5) -- (1.5,7);
\boss{1}{4.5}{0.1};
\boss{1}{3.5}{0.1};
\bosss{1}{2.5}{0.1};
\bosss{1}{1.5}{0.1};
\bosss{2}{2.5}{0.1};
\bosss{2}{1.5}{0.1};
\boss{1.5}{7}{0.1};
\boss{1.7}{7}{0.1};
\bosss{2.5}{7}{0.1};
\bosss{2.7}{7}{0.1};
\end{tikzpicture}
&\quad\quad
\begin{tikzpicture}[scale=0.9]
\foreach\x in {2,...,1}{
\foreach\y in {6,...,1}{
\blank{\x}{\y};
}
\node at (3.5-\x,0.5) {\color{red} \tiny(\x)};
}
\gblank{1}{4};
\gblank{1}{3};
\gblank{1}{2};
\gblank{1}{1};
\gblank{2}{5};
\gblank{2}{1};
\draw[bosssline] (1,1.5) -- (2.7,1.5) -- (2.7,7);
\draw[bosssline] (1,2.5) -- (1.3,2.5) -- (1.3,5.5) -- (2.5,5.5) -- (2.5,7);
\draw[bossline] (1,3.5) -- (1.7,3.5) -- (1.7,7);
\draw[bossline] (1,4.5) -- (1.5,4.5) -- (1.5,7);
\boss{1}{4.5}{0.1};
\boss{1}{3.5}{0.1};
\bosss{1}{2.5}{0.1};
\bosss{1}{1.5}{0.1};
\bosss{2}{5.5}{0.1};
\bosss{2}{1.5}{0.1};
\boss{1.5}{7}{0.1};
\boss{1.7}{7}{0.1};
\bosss{2.5}{7}{0.1};
\bosss{2.7}{7}{0.1};
\end{tikzpicture}
&\quad\quad
\begin{tikzpicture}[scale=0.9]
\foreach\x in {2,...,1}{
\foreach\y in {6,...,1}{
\blank{\x}{\y};
}
\node at (3.5-\x,0.5) {\color{red} \tiny(\x)};
}
\gblank{1}{4};
\gblank{1}{3};
\gblank{1}{2};
\gblank{1}{1};
\gblank{2}{6};
\gblank{2}{1};
\draw[bosssline] (1,1.5) -- (2.7,1.5) -- (2.7,7);
\draw[bosssline] (1,2.5) -- (1.3,2.5) -- (1.3,6.5) -- (2.5,6.5) -- (2.5,7);
\draw[bossline] (1,3.5) -- (1.7,3.5) -- (1.7,7);
\draw[bossline] (1,4.5) -- (1.5,4.5) -- (1.5,7);
\boss{1}{4.5}{0.1};
\boss{1}{3.5}{0.1};
\bosss{1}{2.5}{0.1};
\bosss{1}{1.5}{0.1};
\bosss{2}{6.5}{0.1};
\bosss{2}{1.5}{0.1};
\boss{1.5}{7}{0.1};
\boss{1.7}{7}{0.1};
\bosss{2.5}{7}{0.1};
\bosss{2.7}{7}{0.1};
\end{tikzpicture}
\\
\begin{tikzpicture}[scale=0.9]
\foreach\x in {2,...,1}{
\foreach\y in {6,...,1}{
\blank{\x}{\y};
}
\node at (3.5-\x,0.5) {\color{red} \tiny(\x)};
}
\foreach\y in {6,...,1}{
\node at (0,-\y+7.5) {$x_{\y}$};
}
\gblank{1}{4};
\gblank{1}{3};
\gblank{1}{2};
\gblank{1}{1};
\gblank{2}{5};
\gblank{2}{2};
\draw[bosssline] (1,1.5) -- (1.5,1.5) -- (1.5,2.5) -- (2.7,2.5) -- (2.7,7);
\draw[bosssline] (1,2.5) -- (1.3,2.5) -- (1.3,5.5) -- (2.5,5.5) -- (2.5,7);
\draw[bossline] (1,3.5) -- (1.7,3.5) -- (1.7,7);
\draw[bossline] (1,4.5) -- (1.5,4.5) -- (1.5,7);
\boss{1}{4.5}{0.1};
\boss{1}{3.5}{0.1};
\bosss{1}{2.5}{0.1};
\bosss{1}{1.5}{0.1};
\bosss{2}{5.5}{0.1};
\bosss{2}{2.5}{0.1};
\boss{1.5}{7}{0.1};
\boss{1.7}{7}{0.1};
\bosss{2.5}{7}{0.1};
\bosss{2.7}{7}{0.1};
\end{tikzpicture}
&\quad\quad
\begin{tikzpicture}[scale=0.9]
\foreach\x in {2,...,1}{
\foreach\y in {6,...,1}{
\blank{\x}{\y};
}
\node at (3.5-\x,0.5) {\color{red} \tiny(\x)};
}
\gblank{1}{4};
\gblank{1}{3};
\gblank{1}{2};
\gblank{1}{1};
\gblank{2}{6};
\gblank{2}{2};
\draw[bosssline] (1,1.5) -- (1.5,1.5) -- (1.5,2.5) -- (2.7,2.5) -- (2.7,7);
\draw[bosssline] (1,2.5) -- (1.3,2.5) -- (1.3,6.5) -- (2.5,6.5) -- (2.5,7);
\draw[bossline] (1,3.5) -- (1.7,3.5) -- (1.7,7);
\draw[bossline] (1,4.5) -- (1.5,4.5) -- (1.5,7);
\boss{1}{4.5}{0.1};
\boss{1}{3.5}{0.1};
\bosss{1}{2.5}{0.1};
\bosss{1}{1.5}{0.1};
\bosss{2}{6.5}{0.1};
\bosss{2}{2.5}{0.1};
\boss{1.5}{7}{0.1};
\boss{1.7}{7}{0.1};
\bosss{2.5}{7}{0.1};
\bosss{2.7}{7}{0.1};
\end{tikzpicture}
&\quad\quad
\begin{tikzpicture}[scale=0.9]
\foreach\x in {2,...,1}{
\foreach\y in {6,...,1}{
\blank{\x}{\y};
}
\node at (3.5-\x,0.5) {\color{red} \tiny(\x)};
}
\gblank{1}{4};
\gblank{1}{3};
\gblank{1}{2};
\gblank{1}{1};
\gblank{2}{6};
\gblank{2}{5};
\draw[bosssline] (1,1.5) -- (1.4,1.5) -- (1.4,5.5) -- (2.7,5.5) -- (2.7,7);
\draw[bosssline] (1,2.5) -- (1.2,2.5) -- (1.2,6.5) -- (2.5,6.5) -- (2.5,7);
\draw[bossline] (1,3.5) -- (1.8,3.5) -- (1.8,7);
\draw[bossline] (1,4.5) -- (1.6,4.5) -- (1.6,7);
\boss{1}{4.5}{0.1};
\boss{1}{3.5}{0.1};
\bosss{1}{2.5}{0.1};
\bosss{1}{1.5}{0.1};
\bosss{2}{6.5}{0.1};
\bosss{2}{5.5}{0.1};
\boss{1.6}{7}{0.1};
\boss{1.8}{7}{0.1};
\bosss{2.5}{7}{0.1};
\bosss{2.7}{7}{0.1};
\end{tikzpicture}
\end{tabular}
\end{center}

\noindent
where we highlight in green the tiles which give rise to an $x$ weight (these are the tiles which have an occupied left edge). Note that in column 2 of the lattice it is possible for the type 2 bosons to go around the cylinder, but this does not lead to any particle configurations not already shown here. From these lattice configurations we can read off the product of operators which we need to trace over. We obtain
\begin{align*}
\Omega_{\delta^+}
E_{\delta}(x_1,\dots,x_6;q,t)
=
x_3 x_4 x_5^2 x_6^2 \Tr [k^2 S]
+
x_2  x_3 x_4 x_5 x_6^2 \Tr [a k^2 a^\dag S]
+
x_1  x_3 x_4 x_5 x_6^2 \Tr [a k^2 a^\dag S]
\\
+
x_2 x_3 x_4 x_5^2 x_6 \Tr [a k^2 a^\dag S]
+
x_1 x_3 x_4 x_5^2 x_6 \Tr [a k^2 a^\dag S]
+
x_1 x_2 x_3 x_4 x_5 x_6 \Tr [a a k^2 a^\dag a^\dag S],
\end{align*}
in agreement with our earlier calculation \eqref{eq:Edtrace}.

\section{Symmetric Macdonald polynomials}

We relate our results to symmetric Macdonald polynomials. The following result is already mentioned in \cite{macd95},
\begin{lemma}
\label{le:symmetry}
Let $\lambda$ be a dominant composition, i.e. a partition. Then the sum 
\be
\mathcal{P}_\lambda(x_1,\ldots,x_n;q,t) = \sum_{\mu} f_\mu(x_1,\ldots,x_n;q,t),
\ee 
where the sum runs through all permutations $\mu$ of $\lambda^+$, is symmetric.
\end{lemma}
\begin{proof}
We need to show that $T_i \mathcal{P}_\lambda = t^{1/2} \mathcal{P}_\lambda$ for all $i=1,\ldots,n-1$. From \eref{qKZ2} and \eref{eq:Hecke} we find for $\lambda_i < \lambda_{i+1}$ that
\begin{align}
t^{1/2} T_i f_{\ldots,\lambda_i,\lambda_{i+1},\ldots} &= tT_i^2  f_{\ldots,\lambda_{i+1},\lambda_{i},\ldots} = t \left(1+(t^{1/2}-t^{-1/2})T_i \right)  f_{\ldots,\lambda_{i+1},\lambda_{i},\ldots} \nonumber\\
&= t f_{\ldots,\lambda_{i+1},\lambda_{i},\ldots} +(t-1)  f_{\ldots,\lambda_{i},\lambda_{i+1},\ldots}.
\end{align}
Combining this with  \eref{qKZ1} and \eref{qKZ2} we thus find
\begin{align}
t^{1/2} T_i \sum_\mu f_\mu &= \sum_{\mu:\ \mu_i < \mu_{i+1}} \left( tf_{s_i\mu} +(t-1)f_\mu \right) + \sum_{\mu:\ \mu_i = \mu_{i+1}} t f_\mu + \sum_{\mu:\ \mu_i > \mu_{i+1}} f_{s_i\mu} \nonumber\\
&= \sum_{\mu:\ \mu_i < \mu_{i+1}} tf_{s_i\mu} + \sum_{\mu:\ \mu_i \le \mu_{i+1}} tf_\mu = t \sum_\mu f_\mu.
\end{align}
\end{proof}
The Macdonald polynomial $P_{\lambda}$ is the unique symmetric polynomial (up to normalisation) which can be obtained by taking linear combinations of the non-symmetric Macdonald polynomials $E_{\mu}$, where $\mu$ is a permutation of $\lambda$. By \eqref{eq:triang}, each $f_{\mu}$ can be written as a linear combination of non-symmetric Macdonald polynomials $E_{\nu}$, where $\nu$ is a permutation of $\mu$ and $\nu \le \mu$. It follows that $\mathcal{P}_{\lambda}$ is a linear combination of $E_{\mu}$ such that $\mu$ is a permutation of $\lambda$, and since $\mathcal{P}_{\lambda}$ is symmetric, $\mathcal{P}_{\lambda}$ and $P_{\lambda}$ must be equal up to normalisation, due to uniqueness. By this argument and Lemma~\ref{le:symmetry}, we have the following theorem,

\begin{thm} 
\label{th:symmMacd}
For $\lambda \subset r^n$ 
\be
\Omega_{\lambda}  P_\lambda(x_1,\ldots,x_n;q,t) = \sum_{\mu\in S_n\cdot \lambda} \Tr S \prod_{i=1}^n   A_{\mu_i} (x_i),
\ee
where $\Omega_{\lambda} =\Omega_{\lambda}^{(r)} $, $S=S^{(r)}$ , $A_\mu=A_\mu^{(r)}$ and the sum is over all permutations $\mu$ of $\lambda$.
\end{thm}

A corollary of Theorem~\ref{th:symmMacd} is that the generating function for symmetric Macdonald polynomials for partitions $\lambda\subset r^n$ can be written as a single matrix product. Define first
\begin{equation}
A(y_0,\ldots,y_r;x) = \sum_{i=0}^r y_i A_{i} (x).
\end{equation}
Notice that $[A(y_0,\ldots,y_r;x),A(y_0,\ldots,y_r;x')]=0$ by a similar argument as in Lemma~\ref{le:symmetry}. Moreover we have that
\begin{equation}
\Tr [S\prod_{i=1}^n A(y_0,\ldots,y_r;x_i)] = \sum_{\lambda \subset r^n} \Omega_{\lambda} \prod_{i=0}^r y_i^{m_i(\lambda)} P_\lambda(x_1,\ldots,x_n;q,t), 
\end{equation}
where again all superscripts $(r)$ are suppressed and $m_i(\lambda)$ is the multiplicity of part $i$ in $\lambda$.

Special values of $q$ include the Hall--Littlewood polynomials ($q=0$), which have been expressed as Bethe wave-functions in the $t$-boson model \cite{Tsilevich}. The structure of the expression in \cite{Tsilevich} has some features in common with our matrix product formula, but is distinguished by some important differences. In particular, higher rank $L$-matrices are not needed in the Hall--Littlewood formula, and the construction does not involve a trace operation. Other special values of $q$ include Schur functions ($q=t$), Jack polynomials ($q=t^{\alpha}, t\rightarrow 1$) and $q=1$ results in the normalisation of the inhomogeneous multi-species asymmetric exclusion process with periodic boundary conditions.  The homogeneous limit of the latter is recovered by sending $x_i\rightarrow 1$.

\section{Calculation of traces}
\label{se:traces}

So far we have explained how to calculate the matrix product \eqref{MPA} in terms of traces of 
$t$-oscillators, without explicitly evaluating these. We now turn to the evaluation of the traces, using the representation \eqref{eq:oscillatorsrep} of the $t$-oscillator algebras. Let us start with the examples encountered in equation \eref{eq:Edtrace}:
\be
\begin{split}
\Tr [k^{p}]
&=
\sum_{m=0}^{\infty}
(t^{p})^m
=
\frac{1}{1-t^{p}} = \frac{1}{(1-t)[p]}. 
\\
\frac{
\Tr [a a^\dag k^{p}]
}
{
\Tr [k^p]
}
&=
\frac{
\sum_{m=0}^{\infty}
(t^{p})^{m}
(1-t^{m+1})
}
{
\Tr [k^p]
}
 = \frac{1}{[p+1]},
\\
\frac{
\Tr [a a a^\dag a^\dag k^{p}]
}
{
\Tr [k^p]
}
&=
\frac{
\sum_{m=0}^{\infty}
(t^{p})^m
(1-t^{m+1})
(1-t^{m+2})
}
{
\Tr [k^p]
}
= \frac{[2]}{[p+1][p+2]}.
\end{split}
\ee
In theory, the trace of any product of $t$-oscillators can be evaluated by simply summing up geometric series. On the other hand, as the examples become more complicated, the answer is no longer neatly factorized. In order to evaluate the most general case, let $\mathbb{D}_{\ell}$ be the set of all Dyck paths of length $2\ell$. It is easy to verify that all types of traces arising in this work have the generic form 
$
\Tr [{\rm D}_{\ell} k^{p}]
$
where $p$ is an arbitrary exponent, and ${\rm D}_{\ell}$ represents a Dyck path of length $2\ell$ generated by the oscillators ($a$ means an up step, $a^\dag$ means a down step). Indeed, the trace of any string of $t$-oscillators can be brought into this form (up to overall factors of $t$) using the cyclicity of the trace and the commutation relations \eqref{eq:oscillators}. Let us denote ${\rm D}_{\ell}$ by a series of open/closed parentheses, for example:
\begin{align*}
\Tr [()(())k^p]
:=
\Tr [a a^\dag a a a^\dag a^\dag k^p].
\end{align*}
We consider the following map:
\be
\begin{split}
\mathcal{M}: \quad & \mathbb{D}_{\ell} \rightarrow \mathbb{N}^{\ell}
\\
& {\rm D}_{\ell} \mapsto (m_1,\dots,m_{\ell})
\end{split}
\ee 
where $m_i$ is the number of parenthetic pairs $(\dots)$ in ${\rm D}_{\ell}$ which are surrounded by $i-1$ other parenthetic pairs. For example in the case (()(())) we have $m_1=1$, $m_2=2$ and $m_3=1$. Using this map we deduce the following lemma.
\begin{lemma}
\be
\Tr [{\rm D}_{\ell} k^{p}] = \sum_{n=0}^{\infty}
t^{pn}
(1-t^{n+1})^{m_1}
\dots
(1-t^{n+\ell})^{m_{\ell}},
\ee
where $(m_1,\dots,m_{\ell}) = \mathcal{M}({\rm D}_{\ell})$.
\end{lemma}

It is now useful to introduce the following operators on functions:
\begin{align}
\delta_{t}
[f(z)]
=
f(z)-f(tz),
\quad\quad
\Delta_{t}^{(m)}
=
z
\circ
\underbrace{
\delta_{t}
\circ
\cdots
\circ
\delta_{t}
}_{m},\ \ m\geq 1,
\quad\quad
\Delta_{t}^{(0)}
=
z.
\end{align}
The operator $\Delta_{t}^{(m)}$ acts multiplicatively on monomials and it is easy to see that 
\be
\Delta_{t}^{(m)}(z^n) = (1-t^{n})^m z^{n+1}.
\ee 
It therefore follows that
\begin{multline}
\sum_{n=0}^{\infty}
x^{n}
(1-t^{n+1})^{m_1}
\dots
(1-t^{n+\ell})^{m_{\ell}}\\
=
\left[
\Delta_{t}^{(m_{\ell})}
\cdots
\Delta_{t}^{(m_1)}
\sum_{n=0}^{\infty}
x^n z^{n+1}
\right]_{z=1}
=
\left[
\Delta_{t}^{(m_{\ell})}
\cdots
\Delta_{t}^{(m_1)}
\frac{z}{1-xz}
\right]_{z=1}.
\end{multline}
If we define the function
\begin{align}
\label{eq:psi}
\psi_{[m_1,\dots,m_{\ell}]}
(x)
=
\left[
\Delta_{t}^{(m_{\ell})}
\cdots
\Delta_{t}^{(m_1)}
\frac{z}{1-xz}
\right]_{z=1},
\end{align}
then
\begin{align*}
\Tr [{\rm D}_{\ell} k^p]
=
\psi_{[m_1,\dots,m_{\ell}]}
(t^p),
\quad{\rm where}\ \ 
(m_1,\dots,m_{\ell})
=
\mathcal{M}({\rm D}_{\ell}).
\end{align*}
%
Hence we are able to evaluate all traces that we encounter in terms of the function \eqref{eq:psi}, whose definition is relatively elementary.

\section{Conclusion}
The main result of this paper is a matrix product formula for Macdonald polynomials in terms of deformed bosonic operators. This formula implies a new explicit way to efficiently compute these polynomials and also provides a combinatorial interpretation. Our result firmly connects the polynomial representation theory of the affine Hecke algebra with the theory of solvable lattice models using tools such as the Yang--Baxter equation, $R$-matrices, the Zamolodchikov--Faddeev algebra and the deformed Knizhnik--Zamolodchikov equation. 

The results discussed in this paper have a direct generalisation to the inhomogeneous multi-species asymmetric exclusion process with boundaries \cite{AyyerLS,AyyerLS2,UchiSW,Uchi,CrampeMRV} resulting in a matrix product formula for Koornwinder polynomials \cite{CdGW}. A connection between Koornwinder polynomials and the quantum XXZ chain, which is closely related to the exclusion process by a similarity transformation, was made in \cite{StokmanV}.

We note that a different generalisation of the multi-species asymmetric exclusion process, using inhomogeneous hopping parameters, was discussed in \cite{LamW, AritaM,AyyerL12,AyyerL14}. It would be interesting to clarify the matrix product structure \cite{AritaM} in this case, as well as the relation to integrability \cite{Cantini09}.

\section*{Acknowledgment}
LC has been supported by the CNRS through a Chaire d'Excellence. JdG and MW are generously supported by the Australian Research Council (ARC) and the ARC Centre of Excellence for Mathematical and Statistical Frontiers (ACEMS). We thank Kayed Al Qasemi, Eric Ragoucy, Sergey Sergeev, Ole Warnaar and Paul Zinn-Justin for discussion.

\appendix

\section{A further combinatorial example}

We illustrate the combinatorial meaning of \eqref{trans-eq} on a small rank 3 example, namely $f_{\lambda}(x_1,x_2,x_3,x_4)$ with $\lambda = (3,1,0,2)$. The function $f_{(3,1,0,2)}$ can be expanded into functions $f_\mu$ where $\mu$ is a permuation of $(2,1,0,0)$. So in principle there are 12 terms in this expansion but since the entries (1, 1), (1, 2), and (2, 2) of $\tilde{L}^{(3)}$ vanish, we have only to consider the permutations $(\pi_1,\pi_2,\pi_3,\pi_4)$ of (2, 1, 0, 0) such that $\pi_2 \neq 1, 2$ and  $\pi_4 \neq 2$. Hence we obtain the following four terms in the expansion:
\be
f_{\lambda}=T_{\lambda,(2,0,0,1)}f_{(2,0,0,1)} + T_{\lambda,(0,0,2,1)}f_{(0,0,2,1)} + T_{\lambda,(2,0,1,0)}f_{(2,0,1,0)} + T_{\lambda,(1,0,2,0)}f_{(1,0,2,0)}.
\ee
These terms are represented by the following four pictures:
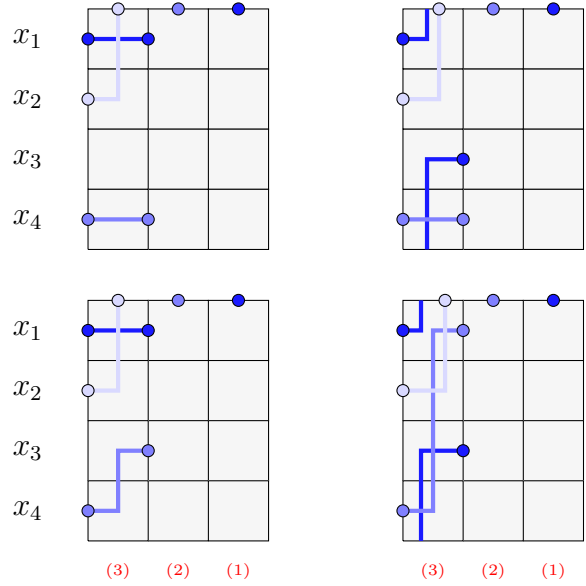
\begin{center}
\begin{tabular}{ccc}
\begin{tikzpicture}[scale=0.8,baseline=(bas.base)]
\node (bas) at (0,3.5) {};
\foreach\x in {3,...,1}{
\foreach\y in {4,...,1}{
\blank{\x}{\y};
}
}
\foreach\y in {4,...,1}{
\node at (0,-\y+5.5) {$x_{\y}$};
}
\draw[bosssline] (1,4.5) -- (2,4.5);
\draw[bossline] (1,1.5) -- (2,1.5);
\draw[bosline] (1,3.5) -- (1.5,3.5) -- (1.5,5);
\boss{1}{1.5}{0.1};
\bos{1}{3.5}{0.1};
\bosss{1}{4.5}{0.1};
\boss{2}{1.5}{0.1};
\bosss{2}{4.5}{0.1};
\bos{1.5}{5}{0.1};
\boss{2.5}{5}{0.1};
\bosss{3.5}{5}{0.1};
\end{tikzpicture}
&
&
\begin{tikzpicture}[scale=0.8,baseline=(bas.base)]
\node (bas) at (0,3.5) {};
\foreach\x in {3,...,1}{
\foreach\y in {4,...,1}{
\blank{\x}{\y};
}
}
\draw[bosssline] (1,4.5) -- (1.4,4.5) -- (1.4,5);
\draw[bosssline] (1.4,1) -- (1.4,2.5) -- (2,2.5);
\draw[bossline] (1,1.5) -- (2,1.5);
\draw[bosline] (1,3.5) -- (1.6,3.5) -- (1.6,5);
\boss{1}{1.5}{0.1};
\bos{1}{3.5}{0.1};
\bosss{1}{4.5}{0.1};
\boss{2}{1.5}{0.1};
\bosss{2}{2.5}{0.1};
\bos{1.6}{5}{0.1};
\boss{2.5}{5}{0.1};
\bosss{3.5}{5}{0.1};
\end{tikzpicture}
\\ \\
\begin{tikzpicture}[scale=0.8,baseline=(bas.base)]
\node (bas) at (0,3.5) {};
\foreach\x in {3,...,1}{
\foreach\y in {4,...,1}{
\blank{\x}{\y};
}
\node at (4.5-\x,0.5) {\color{red}\tiny (\x)};
}
\foreach\y in {4,...,1}{
\node at (0,-\y+5.5) {$x_{\y}$};
}
\draw[bosssline] (1,4.5) -- (2,4.5);
\draw[bossline] (1,1.5) -- (1.5,1.5) -- (1.5,2.5) -- (2,2.5);
\draw[bosline] (1,3.5) -- (1.5,3.5) -- (1.5,5);
\boss{1}{1.5}{0.1};
\bos{1}{3.5}{0.1};
\bosss{1}{4.5}{0.1};
\boss{2}{2.5}{0.1};
\bosss{2}{4.5}{0.1};
\bos{1.5}{5}{0.1};
\boss{2.5}{5}{0.1};
\bosss{3.5}{5}{0.1};
\end{tikzpicture}
&
&
\begin{tikzpicture}[scale=0.8,baseline=(bas.base)]
\node (bas) at (0,3.5) {};
\foreach\x in {3,...,1}{
\foreach\y in {4,...,1}{
\blank{\x}{\y};
}
\node at (4.5-\x,0.5) {\color{red}\tiny (\x)};
}
\draw[bosssline] (1,4.5) -- (1.3,4.5) -- (1.3,5);
\draw[bosssline] (1.3,1) -- (1.3,2.5) -- (2,2.5);
\draw[bossline] (1,1.5) -- (1.5,1.5) -- (1.5,4.5) -- (2,4.5);
\draw[bosline] (1,3.5) -- (1.7,3.5) -- (1.7,5);
\boss{1}{1.5}{0.1};
\bos{1}{3.5}{0.1};
\bosss{1}{4.5}{0.1};
\boss{2}{4.5}{0.1};
\bosss{2}{2.5}{0.1};
\bos{1.7}{5}{0.1};
\boss{2.5}{5}{0.1};
\bosss{3.5}{5}{0.1};
\end{tikzpicture}
\end{tabular}
\end{center}
Replacing each tile in column 3 with the operator(s) it represents, this becomes
\begin{align*}
f_{(3,1,0,2)}
&=
(1-q t) (1-q^{2} t^{2}) x_1 x_2 x_4
\left(
{\rm Tr} [ k_3 k_2 k_3  s ] f_{(2,0,0,1)}
+
{\rm Tr} [ a_3^\dag k_3 k_2 a_3  k_3 s] f_{(0,0,2,1)}
\right. 
\\
&+
\left.
{\rm Tr} [ k_3 k_2 a_2 k_3 a_2^\dag s] f_{(2,0,1,0)} 
+
{\rm Tr} [ a_3^\dag a_2 k_3 k_2 a_3 k_3 a_2^\dag s] f_{(1,0,2,0)}
\right).
\end{align*}
The traces can then be factorized across commuting families of operators:
\begin{align*}
f_{(3,1,0,2)}
&=
(1-q t) (1-q^{2} t^{2}) x_1 x_2 x_4
\left(
{\rm Tr}[k_3^2 s_3]
{\rm Tr}[k_2 s_2]
f_{(2,0,0,1)}
\right.
\\
&+
q^2 t\  
{\rm Tr}[a_3 a_3^\dag k_3^2 s_3]
{\rm Tr}[k_2 s_2]
f_{(0,0,2,1)}
\\
&+
\left. 
{\rm Tr}[k_3^2 s_3]
{\rm Tr}[a_2 a_2^\dag k_2 s_2]
f_{(2,0,1,0)}
+
q^2 t^2\ 
{\rm Tr}[a_3 a_3^\dag k_3^2 s_3]
{\rm Tr}[a_2 a_2^\dag k_2 s_2]
f_{(1,0,2,0)}\right),
\end{align*}
and using the explicit form \eqref{eq:twist-factor} of the twist operators, this becomes
\begin{align*}
f_{(3,1,0,2)}
&=
(1-q t) (1-q^{2} t^{2}) x_1 x_2 x_4
\left(
{\rm Tr}[k^{2+2u} ]
{\rm Tr}[k^{1+u}]
f_{(2,0,0,1)}
\right.
\\
&+
q^2 t\  
{\rm Tr}[a a^\dag k^{2+2u}]
{\rm Tr}[k^{1+u}]
f_{(0,0,2,1)}
\\
&+
\left. 
{\rm Tr}[k^{2+2u}]
{\rm Tr}[a a^\dag k^{1+u}]
f_{(2,0,1,0)}
+
q^2 t^2\ 
{\rm Tr}[a a^\dag k^{2+2u}]
{\rm Tr}[a a^\dag k^{1+u}]
f_{(1,0,2,0)}\right).
\end{align*}
The traces themselves are easily calculated using the results in Section \ref{se:traces}. Thus we obtain the matrix elements $T_{\lambda,(2,0,0,1)}$, $T_{\lambda,(0,0,2,1)}$, $T_{\lambda,(2,0,1,0)}$ and $T_{\lambda,(1,0,2,0)}$.

\end{document}